\documentclass{article}

\usepackage[final]{staix_2026}

\usepackage[utf8]{inputenc}
\usepackage[T1]{fontenc}
\usepackage{hyperref}
\hypersetup{
  colorlinks=true,
  linkcolor=black,
  citecolor=black,
  urlcolor=black
}
\usepackage{url}
\usepackage{booktabs}
\usepackage{amsfonts}
\usepackage{nicefrac}
\usepackage{microtype}
\usepackage{xcolor}
\usepackage{graphicx}
\usepackage{amsmath,amssymb,amsthm}

\usepackage{graphicx} 
\usepackage{amsfonts}
\usepackage{cleveref}
\usepackage{algpseudocode}
\usepackage{float}
\usepackage{algorithm} 
\usepackage{algpseudocode}
\usepackage{booktabs}
\usepackage{natbib}
\usepackage{enumitem}
\usepackage{wrapfig}
\usepackage{url}

\newtheorem{thm}{Theorem}[section]
\newtheorem{prop}{Proposition}[section]

\crefname{thm}{theorem}{theorems}
\Crefname{thm}{Theorem}{Theorems}
\crefname{prop}{proposition}{propositions}
\Crefname{prop}{Proposition}{Propositions}
\crefname{lem}{lemma}{lemmas}
\Crefname{lem}{Lemma}{Lemmas}
\crefname{cor}{corollary}{corollaries}
\Crefname{cor}{Corollary}{Corollaries}

\title{Generative Synthetic Data for Causal Inference: Pitfalls, Remedies, and Opportunities}

\author{%
  Yichen Xu \\
  University of California, Berkeley \\
  \texttt{yichen\_xu@berkeley.edu}
}

\begin{document}

\maketitle

\begin{abstract}
Synthetic tabular data are often evaluated by distributional similarity, privacy distance, or train-on-synthetic-test-on-real predictive performance, but these criteria do not ensure validity for causal inference. We show that fully generative tabular synthesizers, including GAN- and LLM-based models, can preserve predictive utility while distorting average treatment effect (ATE) estimates. The failure is structural: ATE preservation requires both a realistic covariate law and an accurate treatment-effect contrast, whereas prediction loss penalizes treatment-effect error only through an overlap-weighted term. Thus, under imbalance or limited overlap, a generator may reproduce dominant observed outcomes while underlearning intervention-relevant contrasts. We formalize this mismatch through sensitivity and loss-decomposition results. Motivated by this causal analysis and intuition, we propose a hybrid synthetic-data framework for causal inference that generates covariates while modeling treatment and outcome mechanisms separately. We evaluate the framework in three settings: ATE preservation under fully generative versus hybrid synthesis, augmentation for practical positivity problems, and diagnostic simulation engines for comparing OR, IPW, AIPW, and TMLE before real-data analysis. We also stress-test the hybrid construction across settings that vary overlap, covariate dimension, seed sample size, and treatment-effect complexity, including a logistic outcome-model misspecification check. Across controlled simulation experiments, hybrid synthesis improves causal fidelity relative to fully generative baselines; the ACTG application shows improved predictive fidelity and potential for finite-sample estimator benchmarking. LLM-based hybrid synthesis is often more faithful than CTGAN in settings where causal fidelity can be assessed.
\end{abstract}

\section{Introduction}

Synthetic data is increasingly used for data sharing, privacy protection, augmentation, and simulation, but its role in causal inference remains underdeveloped. We consider a basic causal inference setting with covariates \(W\), binary treatment \(A\), and outcome \(Y\). The central difficulty is that causal analysis requires more than reproducing the observed data distribution. For the average treatment effect (ATE), the key object is the conditional treatment-effect contrast
\[
Q(1,W)-Q(0,W),
\qquad
Q(a,W)=\mathbb E[Y\mid A=a,W],
\]
together with the covariate law of \(W\). A synthetic dataset can therefore look realistic, achieve reasonably good train-on-synthetic-test-on-real (TSTR) performance, or remain close to the original data under privacy-distance diagnostics, while still distorting the causal estimand. We formalize this gap through sensitivity and loss-decomposition results showing that predictive fidelity does not imply causal-contrast fidelity.

The failure is structural: predictive objectives primarily reward accuracy on the observed treatment path, whereas causal inference requires preserving the contrast between treatment paths. For the ATE, synthetic validity depends on both the covariate law \(P_W\) and the conditional treatment-effect function \(Q(1,W)-Q(0,W)\). Our sensitivity analysis separates these two sources of error, and our prediction-loss decomposition shows why the second can be weakly controlled by ordinary predictive training. Writing \(m_\pi(W)=\mathbb E[Y\mid W]\) and \(\tau(W)=Q(1,W)-Q(0,W)\), the outcome-prediction loss decomposes into prognostic error plus an overlap-weighted treatment-effect error, with weight \(\pi(W)\{1-\pi(W)\}\). Thus, under imbalance or limited overlap, a model may reproduce dominant factual outcomes while poorly preserving the counterfactual contrast needed for ATE estimation. This explains how fully generative tabular synthesizers, including GAN- and LLM-based models, can perform well on TSTR or reconstruction-style metrics while producing biased causal estimates.

These observations motivate a hybrid generation strategy for tabular data generation. Instead of generating the full triplet \((W,A,Y)\) from one model, we generate covariates \(W\), monitor their realism and privacy using distance-to-closest-record diagnostics, and model treatment and outcome mechanisms separately. When the goal is downstream ATE estimation rather than reproducing the observational treatment policy, the synthetic propensity score is not the central object; treatment assignment may instead be chosen by design, for example randomized, to improve overlap and align outcome learning with the causal contrast. Empirically, across multiple simulation settings that vary overlap, covariate dimension, seed sample size, and treatment-effect complexity, hybrid generation substantially improves causal fidelity for both LLM- and GAN-based generators while keeping privacy-distance diagnostics comparable.

We further study two uses of hybrid synthetic data. First, for practical positivity problems, we introduce synthetic augmentation that pairs extreme-propensity observations with nearby synthetic covariates. This does not solve structural positivity violations, but can reduce error when improved conditional-effect estimation outweighs the covariate shift introduced by augmentation. Second, we develop a hybrid synthetic simulation engine for pre-analysis estimator evaluation. By repeatedly sampling from a learned data-generating environment, analysts can compare OR, IPW, AIPW, and TMLE in terms of finite-sample bias, variance, RMSE, and MSE before the final real-data analysis. Together, these results position synthetic data as a structured component of the causal workflow whose value depends on causal validity, overlap, estimation stability, and diagnostic usefulness.

\section{Related Work}

Relatively few studies examine synthetic data for causal inference. Our work builds on synthetic generation, causal effect estimation, and simulation-based estimator evaluation. Recent work studies LLMs as synthetic data generators for NLP, including issues of generation, curation, faithfulness, diversity, and downstream classification performance~\citep{long24,li23}; related work studies LLM-based synthetic oversampling for imbalanced classification and spurious correlation~\citep{nakada25}, and bias-corrected augmentation when synthetic and target distributions differ~\citep{lyu26}. For tabular data, CTGAN~\citep{xu19} introduced conditional GANs with mode-specific normalization and training-by-sampling, while GReaT~\citep{borisov23} showed that language models can generate realistic tabular rows by serializing them as text. These methods mainly target distributional or predictive fidelity. Several studies consider synthetic data or foundation models for causal inference: \citet{bartolomeis25} use foundation models for randomized experiments, focusing on experimental efficiency; \citet{liu25} survey LLM applications in causal inference; and \citet{amad26} study synthetic data for downstream medical causal inference, proposing causal desiderata, evaluation metrics, and a structured generator.

Robinson-style residualization and the R-learner provide useful context for our decomposition. Robinson's partially linear model uses residualization to separate a target component from nonparametric nuisance structure~\citep{robinson88}. The R-learner extends this idea to heterogeneous treatment-effect estimation by constructing a residualized causal loss: it removes the marginal outcome component and then uses residualized treatment variation to fit the treatment-effect function~\citep{nie21}.

For positivity and limited-overlap problems, many methods reduce instability by trimming, truncating, reweighting, or filtering observations with extreme propensity scores~\citep{crump09,gruber22,xu26tmle,zhou20,joseph07,freedman08,cole08}. Crump et al. characterize subpopulations for which average treatment effects can be estimated most precisely and show that, under homoscedasticity, this leads to trimming observations with extreme propensity scores, for example those below 0.1 or above 0.9. A related line of work uses overlap weighting to emphasize units likely to receive either treatment; these weights are bounded and target the average treatment effect in the overlap population~\citep{li18,cheng22,matsouaka24}.

ATE estimation from observational data is well studied. Outcome regression (OR), or G-computation, estimates \(Q(a,W)=\mathbb E[Y\mid A=a,W]\) and plugs fitted countermeans into the ATE functional~\citep{robins86}; flexible regressors such as BART are often used~\citep{chipman10}. Inverse probability weighting (IPW), related to the Horvitz--Thompson estimator, reweights by the propensity score to emulate randomization~\citep{horvitz52}. Augmented IPW (AIPW) combines outcome and propensity models, giving double robustness if either nuisance model is correct~\citep{robins94}. Targeted maximum likelihood estimation (TMLE) updates an initial outcome regression through a targeted fluctuation submodel to solve the efficient influence-function equation, yielding a doubly robust, locally efficient substitution estimator~\citep{vdl06,vdl11,vdl18}. Because these estimators have distinct finite-sample and limited-overlap failure modes, we benchmark OR, IPW, AIPW, and TMLE.

Causal estimator evaluation often relies on simplified simulations that miss real covariate complexity. ACIC-style benchmarking addresses this limitation by using realistic testing grounds, including real covariates paired with simulated assignment and response-surface mechanisms, to compare many causal estimators~\citep{vincent19}. Generative synthetic data have also been used for bootstrap-style inference. \citet{tran26} use learned generative models to produce synthetic resamples, approximate the sampling distribution of a fixed estimator, and build confidence intervals, including settings where the classical bootstrap can fail.

Our work differs from these related lines in purpose and emphasis. Compared with \citet{amad26}, we analyze ATE preservation directly: our sensitivity result shows that the covariate law and treatment-effect contrast are the essential quantities for ATE preservation, while reproducing the observational treatment-assignment mechanism need not be the primary target when synthetic data are used as a causal-purpose simulation environment. Compared with the R-learner literature, we do not design a causal loss or propose a new CATE estimator; instead, we start from ordinary factual prediction loss and show that it decomposes into prognostic error and an overlap-weighted treatment-effect error. This explains why predictive fidelity can be misleading: synthetic generators may learn dominant covariate-outcome patterns while weakly constraining treatment-effect contrasts in imbalanced or low-overlap regions. Compared with trimming and overlap weighting, our augmentation does not discard observations, truncate weights, or change the estimand to an overlap population; it uses synthetic covariate support and paired treatment assignment to stabilize learning of \(Q(A,W)\) in rare but plausible regimes while accounting for covariate shift. Finally, compared with ACIC-style fixed benchmarks and generative bootstrap methods for one estimator, our hybrid synthetic data serve as a pre-analysis diagnostic simulation engine for comparing OR, IPW, AIPW, and TMLE under a learned finite-sample environment before final real-data analysis.

\section{Failure of Generative Models to Preserve Causal Parameters and a Hybrid Generation Approach}

Synthetic data generation for privacy-preserving data sharing has gained increasing attention, with approaches based on GANs~\citep{xu19} and LLMs~\citep{borisov23}. While metrics such as train-on-synthetic-test-on-real (TSTR) are commonly used to evaluate predictive quality, there is limited focus on whether synthetic data preserve causal parameters. This distinction is important because strong predictive fidelity does not necessarily imply causal fidelity. In particular, the average treatment effect (ATE),
\[
\mathbb E[Y^1-Y^0],
\]
depends on how the outcome changes with treatment conditional on covariates, rather than only on how well a synthetic dataset reproduces marginal or predictive patterns.

We focus on the tuple $(W,A,Y)$, where $W$ denotes covariates, $A\in\{0,1\}$ is a binary treatment, and $Y$ is an outcome. Under standard causal assumptions, including consistency, ignorability, and positivity, the ATE is identified as
\[
\Psi(P)
=
\mathbb E_{P_W}
\left[
Q_P(1,W)-Q_P(0,W)
\right],
\qquad
Q_P(a,w)=\mathbb E_P[Y\mid A=a,W=w].
\]
Thus causal preservation requires preserving both the covariate law $P_W$ and the treatment-effect contrast
\[
\Delta_Q(w):=Q(1,w)-Q(0,w).
\]
The following sensitivity bound makes this point explicit.

\begin{prop}[Sensitivity bound for synthetic ATE]
\label{prop:ate_sensitivity}
Let \(P_W\) and \(P_W^\star\) be covariate laws on a common measurable space
\(\mathcal W\). Let \(Y\in[0,1]\), and define
\[
\Delta_Q(w):=Q(1,w)-Q(0,w),
\qquad
\Delta^\star(w):=Q^\star(1,w)-Q^\star(0,w).
\]
Then
\[
\left|
\Psi(P_W,Q)-\Psi(P_W^\star,Q^\star)
\right|
\le
2\,\mathrm{TV}(P_W,P_W^\star)
+
\|\Delta_Q-\Delta^\star\|_{L_1(P_W^\star)}.
\]
\end{prop}

Proposition~\ref{prop:ate_sensitivity} separates synthetic ATE error into two pieces: an error in the generated covariate law and an error in the outcome contrast. The first term, \(2\,\mathrm{TV}(P_W,P_W^\star)\), measures discrepancy between the target and synthetic covariate laws and applies to continuous, discrete, and mixed covariates. The second term, \(\|\Delta_Q-\Delta^\star\|_{L_1(P_W^\star)}\), measures whether the synthetic outcome mechanism preserves the treatment-effect contrast over the synthetic covariate distribution. This second term is directly relevant for causal estimation, but it is not necessarily prioritized by standard row-level generative objectives. \Cref{thm:joint_reconstruction_tradeoff} provides the intuition that joint reconstruction or next-token objectives can downweight the causal contrast. With \(d\) covariates and one outcome, the outcome loss receives only a \(1/(d+1)\) share, so low joint loss may mostly reflect good covariate reconstruction rather than accurate \(Y\mid A,W\). Since Proposition~\ref{prop:ate_sensitivity} shows that ATE error is directly sensitive to
\[
\|\Delta_Q-\Delta^\star\|_{L_1(P_W^\star)},
\]
a hybrid design can improve ATE preservation by separating covariate generation from outcome modeling and allocating more effort to the contrast-relevant \(Q(Y\mid A,W)\).

More importantly, low prediction risk does not by itself guarantee preservation of the causal parameter, although sufficiently strong prediction-risk control can upper-bound causal-parameter error. The core reason is an objective mismatch between prediction and causal-contrast learning. \Cref{thm:exact_lq_cate_ate_decomposition} makes this mismatch explicit. Let
\[
m_\pi(W)=\mathbb E[Y\mid W],
\qquad
\tau(W)=Q(1,W)-Q(0,W),
\]
denote the prognostic component and the treatment-effect contrast, respectively. For a fitted outcome model, let \(\Delta_m(W)\) and \(\Delta_\tau(W)\) be the corresponding estimation errors. Then the outcome-prediction loss decomposes exactly as
\[
\mathcal L_Q
=
\mathbb E_W[\Delta_m(W)^2]
+
\mathbb E_W\!\left[
\pi(W)\{1-\pi(W)\}\Delta_\tau(W)^2
\right],
\qquad
\pi(W)=\mathbb P(A=1\mid W).
\]
Thus prediction loss is not a pure causal loss: it penalizes error in the prognostic component directly, but penalizes treatment-effect error only through the overlap weight \(\pi(W)\{1-\pi(W)\}\). Even under balanced randomization this weight is \(1/4\), and in observational or limited-overlap settings it can be much smaller. Consequently, a model can predict outcomes accurately while still making large CATE or ATE errors, especially when contrast errors occur in low-overlap regions. An analogous objective-mismatch argument for next-token prediction is given in the Appendix.

Motivated by the loss decomposition, we propose a hybrid data generation approach for tabular causal inference. As shown in \Cref{alg:hybrid}, we first train a generative model on the seed dataset $\mathcal D_{\mathrm{seed}}$ to synthesize covariates $\tilde W$, while monitoring privacy and distributional similarity using distance to closest records (DCR)~\citep{borisov23}. We then construct a treatment-assignment mechanism $\hat g(A\mid W)$ and an outcome model $\hat Q(A,W)$ using $\mathcal D_{\mathrm{seed}}$. Finally, for each synthetic covariate $\tilde W_i$, we sample treatment $\tilde A_i$ from $\hat g(\cdot\mid \tilde W_i)$ and generate outcome $\tilde Y_i$ from $\hat Q(\tilde A_i,\tilde W_i)$.

The treatment mechanism $\hat g$ can either estimate the observational propensity in the seed data or be chosen by design. In particular, when the goal is downstream treatment-effect estimation rather than reproducing the original treatment policy, one may set $\hat g(1\mid W)=1/2$ to create a randomized synthetic experiment. This is appropriate because, under identification, the ATE is governed by the covariate law and the outcome regressions,
\[
\psi=\mathbb E_{\tilde P_W}\{\hat Q(1,W)-\hat Q(0,W)\},
\]
rather than by the synthetic treatment mechanism itself. A randomized $\hat g$ balances treatment arms across the synthetic covariate distribution, making downstream inverse-weighting methods on synthetic data more stable. The separation between the outcome model and the covariate generator also enables safe transfer methods, such as REFINE~\citep{xu26}, to incorporate external information when learning the outcome mechanism.

\begin{algorithm}[H]
\caption{Hybrid synthetic data generation}
\label{alg:hybrid}
\begin{algorithmic}[1]
\State \textbf{Input:} Seed data $\mathcal D_{\mathrm{seed}}=\{(W_i,A_i,Y_i)\}_{i=1}^m$ and target size $n$
\State Construct $\hat P_W$ from $\{W_i\}_{i=1}^m$ \Comment{e.g.\ a generative model, monitored by DCR}
\State Construct $\hat g(a\mid w)$ and $\hat Q(a,w)$ \Comment{trained on $\mathcal D_{\mathrm{seed}}$}
\For{$i=1,\dots,n$}
    \State Sample $\tilde W_i \sim \hat P_W$
    \State Sample $\tilde A_i \sim \hat g(\cdot\mid \tilde W_i)$
    \State Set \(\tilde Y_i \gets \hat Q(\tilde A_i,\tilde W_i)\) \Comment{sample Bernoulli with this mean if \(Y\) is binary}
\EndFor
\State \textbf{Return:} $\widetilde{\mathcal D}=\{(\tilde W_i,\tilde A_i,\tilde Y_i)\}_{i=1}^n$
\end{algorithmic}
\end{algorithm}

We benchmark synthetic data generated by large language models (LLMs), generative adversarial networks (GANs), and their corresponding hybrid variants. In the fully synthetic setting, denoted as "Full generative", both treatment $\tilde A$ and outcome $\tilde Y$ are generated directly from the generative model $p_\theta$. In the hybrid setting, only $\tilde W$ is generated by the synthetic covariate model, while $\tilde A$ and $\tilde Y$ are produced using separately fitted nuisance models.

We report the primary \(d=6\) randomized-seed simulation in the main text and defer additional results that vary dimension, sample size, overlap, and outcome-model specification to \Cref{fig:appendix_rf_outcome_mse,fig:appendix_logistic_outcome_mse,fig:appendix_tstr_dcr_vary}. The main result shows that fully generative synthetic data can preserve useful predictive diagnostics while distorting causal estimands. In \Cref{fig:privacy_parallel}, LLM Full achieves strong TSTR performance, but still yields biased ATE estimates; with true ATE \(0.4183\), its IPW estimate is \(0.4751\) and its TMLE estimate is \(0.5096\). This supports the theoretical point that predictive or reconstruction quality alone does not guarantee preservation of the treatment-effect contrast.

Hybrid construction substantially improves causal fidelity in the main setting across all four estimators. For LLM, MSE decreases from \(0.0047\) to \(0.0004\) for IPW, from \(0.0021\) to \(0.0010\) for AIPW, from \(0.0016\) to \(0.0013\) for OR, and from \(0.0108\) to \(0.0005\) for TMLE. For GAN, MSE decreases from \(0.1854\) to \(0.0114\) for IPW, from \(0.1851\) to \(0.0180\) for AIPW, from \(0.1826\) to \(0.0197\) for OR, and from \(0.2027\) to \(0.0087\) for TMLE. The DCR values remain comparable within each generator family, while the hybrid design improves the estimand-relevant component of the data-generating process by modeling the outcome mechanism separately.

The appendix stress tests support the same conclusion beyond the primary setting. Under the outcome nuisance specification, hybrid synthesis reduces ATE MSE across heterogeneous-effect settings that vary overlap, dimension, and seed sample size. For example, in the \(d=20\), \(n_{\mathrm{seed}}=1000\), poor-overlap complex setting, LLM Hybrid reduces MSE from \(0.0115\) to \(0.0002\) for IPW, from \(0.0017\) to \(0.0003\) for AIPW, from \(0.0016\) to \(0.0003\) for OR, and from \(0.0068\) to \(0.0003\) for TMLE. In the \(d=20\), \(n_{\mathrm{seed}}=500\), poor-overlap complex setting, LLM Hybrid similarly reduces OR MSE from \(0.0243\) to \(0.0035\), while AIPW MSE drops from \(0.0238\) to \(0.0034\). For GAN, the gains are especially large in the \(d=6\), poor-overlap complex setting: IPW MSE decreases from \(0.1797\) to \(0.0025\), AIPW from \(0.1789\) to \(0.0025\), OR from \(0.1784\) to \(0.0025\), and TMLE from \(0.1786\) to \(0.0025\). The logistic outcome-model misspecification check further shows that these improvements are not only due to alignment between the hybrid outcome generator and a flexible outcome learner: when OR, AIPW, and TMLE use logistic outcome models, hybrid synthesis still gives the largest gains in complex or limited-overlap settings. In the simple treatment-effect setting, fully generative LLM synthesis can be competitive; for example, LLM Full achieves lower MSE than LLM Hybrid for OR (\(0.0011\) versus \(0.0037\)) and TMLE (\(0.0005\) versus \(0.0037\)). This suggests that hybridization is most useful when the treatment-effect contrast is difficult to learn from fully generative row-level objectives.

\begin{figure*}[t]
    \centering
    \begin{minipage}{0.49\textwidth}
        \centering
        \includegraphics[width=\textwidth]{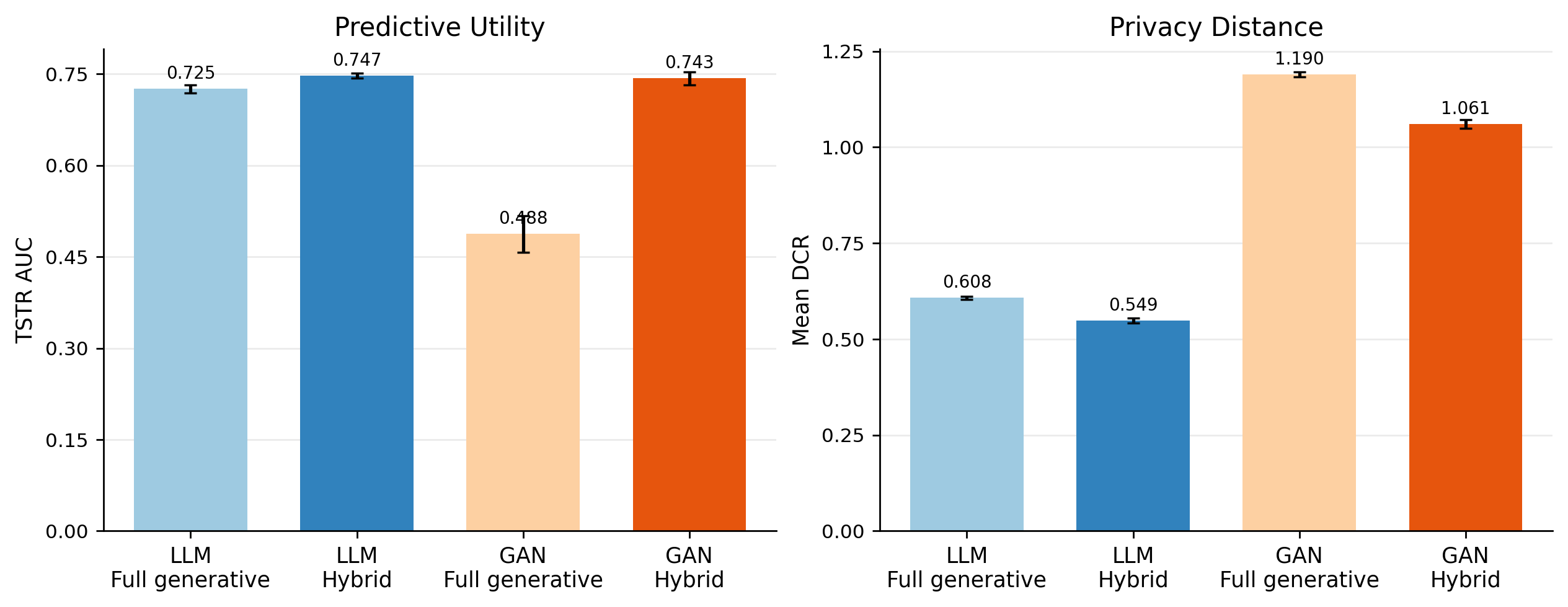}
    \end{minipage}
    \hfill
    \begin{minipage}{0.49\textwidth}
        \centering
        \includegraphics[width=\textwidth]{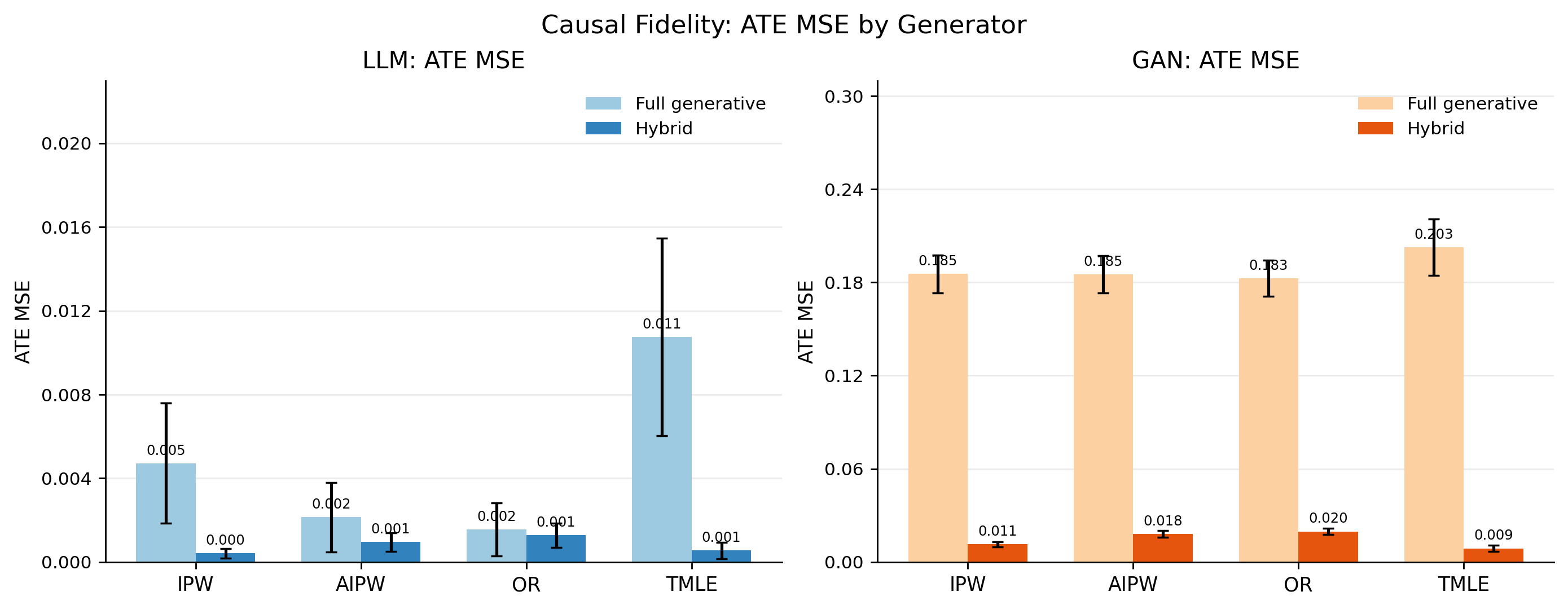}
    \end{minipage}
    \caption{Privacy and causal-fidelity diagnostics for synthetic data. Left: predictive utility (TSTR AUC) and privacy distance (mean DCR) across synthetic datasets. Right: ATE MSE across estimators, shown separately for LLM-based and GAN-based synthetic data. Hybrid constructions improve causal fidelity substantially relative to fully generative synthetic data, while DCR and TSTR alone do not fully determine causal usefulness.}
    \label{fig:privacy_parallel}
\end{figure*}

\section{Synthetic Augmentation for Practical Positivity Problems}

Positivity violations make causal estimation unstable because some treatment values are rarely observed in parts of the covariate space. In these rare $(A,W)$ regimes, the outcome regression $Q(A,W)$ can be poorly learned, and weighting-based estimators can become unstable due to extreme propensity scores. Synthetic data may help by providing plausible covariate scaffolds in sparse regions and by assigning treatments to improve balance over $(W,A)$. This does not solve structural positivity violations, but it can help when the observed data are imbalanced while the target population still contains plausible underrepresented regimes.

\Cref{prop:overlap} formalizes the tradeoff. Synthetic augmentation can reduce error by improving the conditional treatment-effect estimate, but it may also introduce error by shifting the covariate distribution.

Motivated by Proposition~\ref{prop:overlap}, we use synthetic data to support the poorly represented parts of the covariate-treatment space. Let \(p_i=\hat g(1\mid W_i)\) denote the estimated propensity score for observation \(i\). We define samples satisfying
\[
p_i < \frac{1}{\sqrt{n}\log n}
\quad\text{or}\quad
p_i > 1-\frac{1}{\sqrt{n}\log n}
\]
as extreme-propensity samples. For these samples, we search for synthetic counterparts that are close in covariate space according to Euclidean distance and pair them to improve treatment balance.

\begin{table*}[t]
\centering
\caption{MSE across positivity experiments for $n=200$. Lower is better. Values are mean $\pm$ SE across five seeds.}
\label{tab:positivity_mse_n200}
\resizebox{\textwidth}{!}{%
\begin{tabular}{lcccc}
\toprule
Scenario & IPW & AIPW & OR & TMLE \\
\midrule
Original & 0.0082 $\pm$ 0.0014 & 0.0232 $\pm$ 0.0105 & 0.0259 $\pm$ 0.0110 & 0.0104 $\pm$ 0.0035 \\

Pair Hybrid GAN & 0.0025 $\pm$ 0.0015 & 0.0034 $\pm$ 0.0018 & 0.0045 $\pm$ 0.0024 & 0.0057 $\pm$ 0.0034 \\
Pair Hybrid LLM & 0.0010 $\pm$ 0.0003 & 0.0020 $\pm$ 0.0013 & 0.0035 $\pm$ 0.0017 & 0.0042 $\pm$ 0.0015 \\
Pair Self-Supervised GAN & 0.0247 $\pm$ 0.0104 & 0.0384 $\pm$ 0.0146 & 0.0432 $\pm$ 0.0149 & 0.0190 $\pm$ 0.0109 \\
Pair Self-Supervised LLM & 0.0197 $\pm$ 0.0058 & 0.0308 $\pm$ 0.0084 & 0.0351 $\pm$ 0.0096 & 0.0111 $\pm$ 0.0042 \\

Pair Hybrid Flip 5\% GAN & 0.0043 $\pm$ 0.0033 & 0.0091 $\pm$ 0.0061 & 0.0115 $\pm$ 0.0068 & 0.0045 $\pm$ 0.0020 \\
Pair Hybrid Flip 5\% LLM & 0.0029 $\pm$ 0.0021 & 0.0072 $\pm$ 0.0044 & 0.0090 $\pm$ 0.0046 & 0.0045 $\pm$ 0.0019 \\
Pair Hybrid Flip 10\% GAN & 0.0063 $\pm$ 0.0023 & 0.0129 $\pm$ 0.0039 & 0.0157 $\pm$ 0.0045 & 0.0044 $\pm$ 0.0017 \\
Pair Hybrid Flip 10\% LLM & 0.0043 $\pm$ 0.0041 & 0.0102 $\pm$ 0.0076 & 0.0125 $\pm$ 0.0081 & 0.0061 $\pm$ 0.0037 \\
\bottomrule
\end{tabular}%
}
\end{table*}

The results in \Cref{tab:positivity_mse_n200} support this tradeoff. Pair Hybrid LLM reduces MSE relative to the original data for all four estimators: IPW decreases from \(0.0082\) to \(0.0010\), AIPW from \(0.0232\) to \(0.0020\), OR from \(0.0259\) to \(0.0035\), and TMLE from \(0.0104\) to \(0.0042\). Pair Hybrid GAN also improves over the original data, reducing IPW MSE to \(0.0025\), AIPW to \(0.0034\), OR to \(0.0045\), and TMLE to \(0.0057\). In contrast, self-supervised pairing, i.e., training \(Q\) only on the current seed data and then predicting outcomes over synthetic covariates \(W\), is less reliable and can perform worse than the original data. This suggests that synthetic covariates alone do not resolve the positivity problem: without additional outcome information or improved transfer for learning \(Q(A,W)\) in rare regimes, augmentation can fail to improve causal estimation.

The flip experiments, in which we flip the binary outcome labels used to train the outcome model for hybrid synthetic data generation, show that the benefit does not require a perfect synthetic outcome model. With $10\%$ flips, Pair Hybrid LLM still improves over the original data for all four estimators, with MSEs $0.0043$ for IPW, $0.0102$ for AIPW, $0.0125$ for OR, and $0.0061$ for TMLE. Pair Hybrid GAN with $10\%$ flips also remains below the original-data MSE for all four estimators. Additional sample-size results in \Cref{tab:positivity_mse_n100,tab:positivity_mse_n500} show the same qualitative pattern at $n=100$ and $n=500$: hybrid paired augmentation generally improves causal estimation, while self-supervised pairing is less stable. These results are consistent with Proposition~\ref{prop:overlap}: augmentation helps when the improvement in treatment-contrast learning outweighs the covariate shift or outcome noise introduced by the synthetic samples.

\section{Synthetic data as a diagnostic simulation engine for estimator evaluation}

Another use of hybrid synthetic data is estimator evaluation. In practice, analysts often need guidance before committing to a final causal analysis: which estimator is likely to be stable, whether bias or variance will dominate the error, whether a sample size such as 1000 is adequate, and whether OR, IPW, AIPW, or TMLE is safer in the observed regime. Real datasets rarely provide repeated experiments with known ground truth, while hand-crafted simulations often miss the complexity of real covariate distributions. A simulation engine powered by generative synthetic data offers a middle ground: it turns one observed dataset into a realistic finite-sample forecasting environment. We build this engine by learning a realistic covariate generator and separately learning treatment and outcome mechanisms. Once a large synthetic population is generated, we estimate a high-precision reference effect using a large synthetic sample with the doubly robust and semiparametrically efficient method TMLE and then repeatedly subsample smaller datasets to evaluate finite-sample bias, variance, RMSE, and MSE. We report the primary simulation-engine result in the main text. Additional MSE ranking diagnostics across simulation settings are reported in \Cref{tab:sim_engine_rank_agreement}.

In our experiment, the large synthetic reference sample has size 50,000, and each finite-sample replication uses 1,000 observations. \Cref{tab:real_fidelity_summary} compares the finite-sample behavior induced by synthetic data against the corresponding real-data benchmark. This comparison is intended as a diagnostic of estimator behavior in the learned synthetic environment, not as independent validation of the real causal effect. The LLM-based hybrid simulator tracks the real benchmark reasonably well: the signs of bias are correct for all estimators, and the synthetic biases remain close to the real biases. For example, IPW has real bias $-0.0012$ and synthetic bias $-0.0029$, while TMLE has real bias $0.0243$ and synthetic bias $0.0202$. AIPW and OR also show similar bias patterns, with synthetic biases $-0.0302$ and $-0.0375$ compared with real biases $-0.0263$ and $-0.0343$. Although variance magnitudes are not perfectly matched, the LLM simulator still captures informative estimator-ranking patterns. The GAN-based simulator is less faithful: GAN synthetic IPW has bias $-0.1025$ compared with real bias $-0.0012$, GAN synthetic AIPW has bias $-0.1275$ compared with real bias $-0.0263$, and TMLE even has the wrong bias sign. These discrepancies are also reflected in RMSE and MSE, where GAN synthetic errors are much larger than the real-data benchmarks for several estimators.

The results suggest that hybrid synthetic data can serve as a practical diagnostic tool for estimator selection, but the quality of the covariate generator matters. The LLM-based simulator more closely reproduces the real finite-sample ranking and error decomposition, making it more useful for anticipating estimator behavior before the final analysis. 

\begin{table*}[t!]
\centering
\caption{Real-data fidelity diagnostic summary. Synthetic performance is evaluated by how closely synthetic-data finite-sample behavior matches the real-data benchmark, using synthetic-versus-real bias, variance, RMSE, and MSE. MSE values are reported as mean $\pm$ SE.}
\label{tab:real_fidelity_summary}
\scriptsize
\setlength{\tabcolsep}{4pt}
\renewcommand{\arraystretch}{1.15}
\resizebox{\textwidth}{!}{%
\begin{tabular}{llccccccccc}
\hline
Source & Estimator & Sign correct & Real bias & Syn.\ bias & Real var & Syn.\ var & Real RMSE & Syn.\ RMSE & Real MSE & Syn.\ MSE \\
\hline
LLM & IPW  & Yes & -0.0012 & -0.0029 & 0.000279 & 0.000660 & 0.0163 & 0.0252 & 0.000267 $\pm$ 0.000065 & 0.000635 $\pm$ 0.000389 \\
LLM & TMLE & Yes &  0.0243 &  0.0202 & 0.000379 & 0.000762 & 0.0308 & 0.0337 & 0.000948 $\pm$ 0.000204 & 0.001133 $\pm$ 0.000279 \\
LLM & AIPW & Yes & -0.0263 & -0.0302 & 0.000314 & 0.000621 & 0.0315 & 0.0388 & 0.000989 $\pm$ 0.000222 & 0.001504 $\pm$ 0.000623 \\
LLM & OR   & Yes & -0.0343 & -0.0375 & 0.000315 & 0.000630 & 0.0384 & 0.0447 & 0.001478 $\pm$ 0.000277 & 0.002002 $\pm$ 0.000687 \\
\hline
GAN & IPW  & Yes & -0.0012 & -0.1025 & 0.000279 & 0.000530 & 0.0163 & 0.1050 & 0.000267 $\pm$ 0.000065 & 0.011015 $\pm$ 0.001016 \\
GAN & TMLE & No  &  0.0243 & -0.0808 & 0.000379 & 0.000892 & 0.0308 & 0.0859 & 0.000948 $\pm$ 0.000204 & 0.007380 $\pm$ 0.001103 \\
GAN & AIPW & Yes & -0.0263 & -0.1275 & 0.000314 & 0.000589 & 0.0315 & 0.1297 & 0.000989 $\pm$ 0.000222 & 0.016822 $\pm$ 0.001417 \\
GAN & OR   & Yes & -0.0343 & -0.1343 & 0.000315 & 0.000578 & 0.0384 & 0.1363 & 0.001478 $\pm$ 0.000277 & 0.018582 $\pm$ 0.001488 \\
\hline
\end{tabular}%
}
\end{table*}

\section{Real-World Applications and Limitations}

We next evaluate our framework on the ACTG175 dataset, a real-world HIV randomized trial with a continuous outcome. Unlike the earlier privacy experiment, where the target effect is known by construction, the ACTG setting is intended to assess whether synthetic data can support two realistic goals simultaneously: preserving useful predictive and causal structure, and providing a plausible simulation environment for diagnosing estimators under complex covariates in finite samples.

We construct four synthetic ACTG datasets: LLM Full, LLM Hybrid, CTGAN Full, and CTGAN Hybrid. For each synthetic dataset, we report three diagnostics: train-on-synthetic-test-on-real (TSTR) RMSE for the continuous \texttt{cd420} outcome, mean distance to the closest record (DCR), and repeated treatment-effect estimates based on subsamples of size 1000. The resulting summary is shown in \Cref{tab:actg_tstr_ate_dcr}. On predictive fidelity, both hybrid constructions improve over their fully generative counterparts: TSTR RMSE decreases from \(129.73\) to \(102.30\) for the LLM generator and from \(244.46\) to \(107.95\) for CTGAN. The DCR summaries remain similar within each generator family, with mean DCR near \(1.05\) for LLM Full and LLM Hybrid and around \(2.09\) and \(2.00\) for CTGAN Full and CTGAN Hybrid. The treatment-effect estimates also change after hybridization: for both LLM and CTGAN, the hybrid datasets produce more consistently negative ATE estimates across IPW, AIPW, and TMLE. 

\begin{table*}[t]
\centering
\caption{ACTG synthetic-data diagnostics and effect estimates. TSTR RMSE reports train-on-synthetic-test-on-real prediction error for the continuous \texttt{cd420} outcome on the real ACTG benchmark. DCR summarizes privacy distance, and ATE columns report mean estimated treatment effects over repeated subsamples of size 1000.}
\label{tab:actg_tstr_ate_dcr}
\begin{tabular}{lccccc}
\toprule
Dataset & TSTR RMSE & Mean DCR & IPW ATE & AIPW ATE & TMLE ATE \\
\midrule
LLM Full      & 129.73 & 1.0479 &  -3.1670 & -10.0869 &  -2.9390 \\
LLM Hybrid    & 102.30 & 1.0461 & -20.4772 & -14.6497 & -23.9624 \\
CTGAN Full    & 244.46 & 2.0900 &  -2.4600 & -18.7035 &  -1.4010 \\
CTGAN Hybrid  & 107.95 & 1.9969 & -13.5488 & -27.0269 & -17.5316 \\
\bottomrule
\end{tabular}
\end{table*}

Beyond point estimates, the ACTG experiment assesses whether hybrid synthetic data can serve as a realistic simulation engine for finite-sample benchmarking. We generate large hybrid pools, repeatedly subsample datasets of size 100--1000, and evaluate IPW, AIPW, outcome regression, and TMLE relative to a large-sample synthetic TMLE reference within the same hybrid environment. The resulting bias, variance, MSE, and RMSE curves are reported in \Cref{fig:actg_llm_2x2} and \Cref{fig:actg_ctgan_2x2} in the Appendix. Several consistent patterns emerge: IPW has relatively small bias but the largest variance, reflecting instability of weight-based estimators under estimated propensities; biases generally decrease with sample size relative to the large-sample TMLE reference; and AIPW usually improves over outcome regression in both MSE and bias, highlighting the benefit of doubly robust correction. As a calibration check, the Appendix shows that AIPW and TMLE are close on the large-sample LLM hybrid pool, as expected from their shared asymptotic target; this supports the synthetic world as a coherent benchmarking environment and suggests that their differences at \(n=1000\) reflect genuine finite-sample behavior. The CTGAN-based simulator shows similar qualitative trends but more irregular bias and MSE behavior than the LLM-based simulator. Overall, hybrid synthetic data show potential to forecast estimator behavior in realistic regimes where practitioners need guidance about finite-sample estimator behavior before final analysis. Because ACTG has no known ground-truth ATE, these trends should be interpreted as diagnostic evidence about the learned synthetic environment rather than confirmatory evidence of causal accuracy.

These ACTG results illustrate both the promise and the limitation of our framework. The promise is that hybrid synthetic data can convert a single real dataset into a realistic benchmarking environment that reveals estimator-specific failure modes before a final analysis is committed. The limitation is that the usefulness of this benchmark still depends critically on the quality of the learned synthetic covariate generator. In our experiments, the LLM-based hybrid simulator appears substantially more faithful than the CTGAN-based one, suggesting that not all synthetic generators are equally suitable as foundations for causal simulation. Thus, while hybrid synthetic data can support estimator evaluation in complex real-world settings, careful validation of the synthetic data generator remains essential.

\section{Conclusion}

Synthetic data for causal inference should be judged by the causal role it plays in the analysis, not only by global realism, privacy distance, or predictive utility. The central pitfall is that row-level generative objectives can preserve dominant factual patterns while leaving the treatment-effect contrast weakly constrained. This can happen when covariate reconstruction dominates the joint loss, when factual outcome prediction is largely driven by prognostic covariate-outcome structure, or when treatment assignment has limited overlap. In such cases, low joint or predictive loss may provide insufficient control of the contrast \(Q(1,W)-Q(0,W)\) that defines the ATE.

A corresponding remedy is to make outcome and contrast learning explicit. One route is to require sufficiently small factual outcome loss, so that the implicit overlap-weighted contrast component is also controlled. The hybrid strategy separates covariate generation from treatment and outcome modeling, allowing the outcome mechanism to be trained and diagnosed as an estimand-relevant component rather than as a small part of a full-row generative objective. Designed synthetic treatment assignment, such as randomized assignment, can further improve overlap in the synthetic environment and make outcome learning more aligned with the causal contrast needed for ATE estimation.

This also creates opportunities for causal workflows. Hybrid synthetic data can augment rare but plausible treatment-covariate regimes when the gain in contrast learning outweighs the induced covariate shift. It can also turn a single observed dataset into a learned simulation environment for pre-analysis estimator diagnostics, helping compare finite-sample behavior of OR, IPW, AIPW, and TMLE before final analysis. The present study focuses on binary-treatment ATE problems with fitted nuisance mechanisms. Future work should extend the same principle to multiple or continuous treatments, longitudinal regimes, stronger privacy constraints, and broader modern tabular generators. Another promising direction is to connect this tabular causal perspective with pretraining-data selection for LLMs, where data should also be selected or generated not only for global predictive fit, but for preserving the contrasts and subpopulation structure relevant to downstream decisions.

\bibliographystyle{plainnat}
\bibliography{staix_bib}

\clearpage 
\newpage
\appendix

\setcounter{figure}{0}
\renewcommand{\thefigure}{A.\arabic{figure}}

\setcounter{table}{0}
\renewcommand{\thetable}{A.\arabic{table}}

\setcounter{thm}{0}
\renewcommand{\thethm}{A.\arabic{thm}}

\setcounter{prop}{0}
\renewcommand{\theprop}{A.\arabic{prop}}

\setcounter{lem}{0}
\renewcommand{\thelem}{A.\arabic{lem}}

\setcounter{cor}{0}
\renewcommand{\thecor}{A.\arabic{cor}}

\section*{Appendix}
\label{sec:app}

\section*{Additional experiments and reproducibility details}

\subsection*{Experiments}

Tables~\ref{tab:positivity_mse_n100} and~\ref{tab:positivity_mse_n500} report additional positivity experiments at \(n=100\) and \(n=500\). Table~\ref{tab:sim_engine_rank_agreement} summarizes estimator-ranking agreement between real and synthetic finite-sample behavior across simulation settings. Figures~\ref{fig:appendix_rf_outcome_mse}--\ref{fig:appendix_tstr_dcr_vary} report additional ATE MSE, misspecification, TSTR, and DCR diagnostics used to stress-test the hybrid construction.

\Cref{fig:actg_llm_2x2,fig:actg_ctgan_2x2} provide the finite-sample ACTG simulation-engine curves corresponding to the main-text discussion. The LLM- and CTGAN-based hybrid simulators are evaluated by bias, variance, MSE, and RMSE across sample sizes for IPW, AIPW, outcome regression, and TMLE.

\begin{table*}[t]
\centering
\caption{MSE across positivity experiments for $n=100$. Lower is better. Values are mean $\pm$ SE across five seeds.}
\label{tab:positivity_mse_n100}
\resizebox{\textwidth}{!}{%
\begin{tabular}{lcccc}
\toprule
Scenario & IPW & AIPW & OR & TMLE \\
\midrule
Original & 0.0107 $\pm$ 0.0044 & 0.0438 $\pm$ 0.0181 & 0.0458 $\pm$ 0.0182 & 0.0252 $\pm$ 0.0213 \\

Pair Hybrid GAN & 0.0104 $\pm$ 0.0047 & 0.0169 $\pm$ 0.0105 & 0.0217 $\pm$ 0.0132 & 0.0130 $\pm$ 0.0059 \\
Pair Hybrid LLM & 0.0070 $\pm$ 0.0027 & 0.0116 $\pm$ 0.0025 & 0.0160 $\pm$ 0.0046 & 0.0078 $\pm$ 0.0077 \\
Pair Self-Supervised GAN & 0.0408 $\pm$ 0.0131 & 0.0608 $\pm$ 0.0145 & 0.0715 $\pm$ 0.0146 & 0.0248 $\pm$ 0.0111 \\
Pair Self-Supervised LLM & 0.0518 $\pm$ 0.0312 & 0.0585 $\pm$ 0.0266 & 0.0634 $\pm$ 0.0255 & 0.0450 $\pm$ 0.0312 \\

Pair Hybrid Flip 5\% GAN & 0.0077 $\pm$ 0.0053 & 0.0139 $\pm$ 0.0081 & 0.0185 $\pm$ 0.0096 & 0.0080 $\pm$ 0.0025 \\
Pair Hybrid Flip 5\% LLM & 0.0088 $\pm$ 0.0018 & 0.0144 $\pm$ 0.0066 & 0.0190 $\pm$ 0.0090 & 0.0113 $\pm$ 0.0058 \\
Pair Hybrid Flip 10\% GAN & 0.0202 $\pm$ 0.0136 & 0.0352 $\pm$ 0.0171 & 0.0427 $\pm$ 0.0188 & 0.0128 $\pm$ 0.0098 \\
Pair Hybrid Flip 10\% LLM & 0.0066 $\pm$ 0.0039 & 0.0134 $\pm$ 0.0061 & 0.0193 $\pm$ 0.0075 & 0.0055 $\pm$ 0.0033 \\
\bottomrule
\end{tabular}%
}
\end{table*}

\begin{table*}[t]

\centering
\caption{MSE across positivity experiments for $n=500$. Lower is better. Values are mean $\pm$ SE across five seeds.}
\label{tab:positivity_mse_n500}
\resizebox{\textwidth}{!}{%
\begin{tabular}{lcccc}
\toprule
Scenario & IPW & AIPW & OR & TMLE \\
\midrule
Original & 0.0031 $\pm$ 0.0017 & 0.0287 $\pm$ 0.0122 & 0.0311 $\pm$ 0.0126 & 0.0061 $\pm$ 0.0034 \\

Pair Hybrid GAN & 0.0030 $\pm$ 0.0009 & 0.0058 $\pm$ 0.0024 & 0.0071 $\pm$ 0.0028 & 0.0031 $\pm$ 0.0012 \\
Pair Hybrid LLM & 0.0023 $\pm$ 0.0012 & 0.0047 $\pm$ 0.0020 & 0.0062 $\pm$ 0.0025 & 0.0015 $\pm$ 0.0006 \\
Pair Self-Supervised GAN & 0.0359 $\pm$ 0.0142 & 0.0468 $\pm$ 0.0164 & 0.0508 $\pm$ 0.0169 & 0.0284 $\pm$ 0.0135 \\
Pair Self-Supervised LLM & 0.0299 $\pm$ 0.0127 & 0.0384 $\pm$ 0.0142 & 0.0415 $\pm$ 0.0147 & 0.0244 $\pm$ 0.0117 \\

Pair Hybrid Flip 5\% GAN & 0.0059 $\pm$ 0.0031 & 0.0101 $\pm$ 0.0037 & 0.0121 $\pm$ 0.0041 & 0.0025 $\pm$ 0.0016 \\
Pair Hybrid Flip 5\% LLM & 0.0056 $\pm$ 0.0020 & 0.0102 $\pm$ 0.0028 & 0.0126 $\pm$ 0.0033 & 0.0021 $\pm$ 0.0009 \\
Pair Hybrid Flip 10\% GAN & 0.0080 $\pm$ 0.0021 & 0.0135 $\pm$ 0.0023 & 0.0160 $\pm$ 0.0025 & 0.0035 $\pm$ 0.0011 \\
Pair Hybrid Flip 10\% LLM & 0.0067 $\pm$ 0.0030 & 0.0108 $\pm$ 0.0035 & 0.0127 $\pm$ 0.0038 & 0.0034 $\pm$ 0.0019 \\
\bottomrule
\end{tabular}%
}
\end{table*}

\begin{table*}[t!]
\centering
\caption{Estimator MSE ranking agreement between real and synthetic finite-sample behavior. For each setting and source, the table reports the real-data MSE rank, synthetic-data MSE rank, and pairwise rank agreement across the four estimators. Lower MSE is better. Pairwise accuracy is computed over the six estimator pairs per setting.}
\label{tab:sim_engine_rank_agreement}
\scriptsize
\setlength{\tabcolsep}{4pt}
\renewcommand{\arraystretch}{1.15}
\resizebox{\textwidth}{!}{%
\begin{tabular}{llccll}
\hline
Setting & Source & Correct pairs & Accuracy & Real MSE rank & Synthetic MSE rank \\
\hline
$d=20$, poor overlap, complex outcome & LLM & 6/6 & 100.0\% & IPW $<$ TMLE $<$ AIPW $<$ OR & IPW $<$ TMLE $<$ AIPW $<$ OR \\
 & GAN & 5/6 & 83.3\% & IPW $<$ TMLE $<$ AIPW $<$ OR & TMLE $<$ IPW $<$ AIPW $<$ OR \\
\hline

$d=20$, poor overlap, simple outcome & LLM & 5/6 & 83.3\% & IPW $<$ TMLE $<$ AIPW $<$ OR & TMLE $<$ IPW $<$ AIPW $<$ OR \\
 & GAN & 5/6 & 83.3\% & IPW $<$ TMLE $<$ AIPW $<$ OR & TMLE $<$ IPW $<$ AIPW $<$ OR \\
\hline

$d=20$, poor overlap, complex outcome, $n_{\mathrm{seed}}=500$ & LLM & 5/6 & 83.3\% & IPW $<$ TMLE $<$ AIPW $<$ OR & IPW $<$ AIPW $<$ TMLE $<$ OR \\
 & GAN & 5/6 & 83.3\% & IPW $<$ TMLE $<$ AIPW $<$ OR & TMLE $<$ IPW $<$ AIPW $<$ OR \\
\hline

$d=6$, moderate overlap, complex outcome & LLM & 5/6 & 83.3\% & AIPW $<$ OR $<$ IPW $<$ TMLE & OR $<$ AIPW $<$ IPW $<$ TMLE \\
 & GAN & 2/6 & 33.3\% & AIPW $<$ OR $<$ IPW $<$ TMLE & IPW $<$ TMLE $<$ AIPW $<$ OR \\
\hline

$d=6$, poor overlap, complex outcome & LLM & 5/6 & 83.3\% & AIPW $<$ OR $<$ IPW $<$ TMLE & OR $<$ AIPW $<$ IPW $<$ TMLE \\
 & GAN & 6/6 & 100.0\% & AIPW $<$ OR $<$ IPW $<$ TMLE & AIPW $<$ OR $<$ IPW $<$ TMLE \\
\hline

\textbf{Overall} & \textbf{LLM} & \textbf{26/30} & \textbf{86.7\%} & -- & -- \\
\textbf{Overall} & \textbf{GAN} & \textbf{23/30} & \textbf{76.7\%} & -- & -- \\
\hline
\end{tabular}%
}
\end{table*}

\begin{figure*}[t]
\centering

\begin{minipage}{0.48\textwidth}
    \centering
    \includegraphics[width=\textwidth]{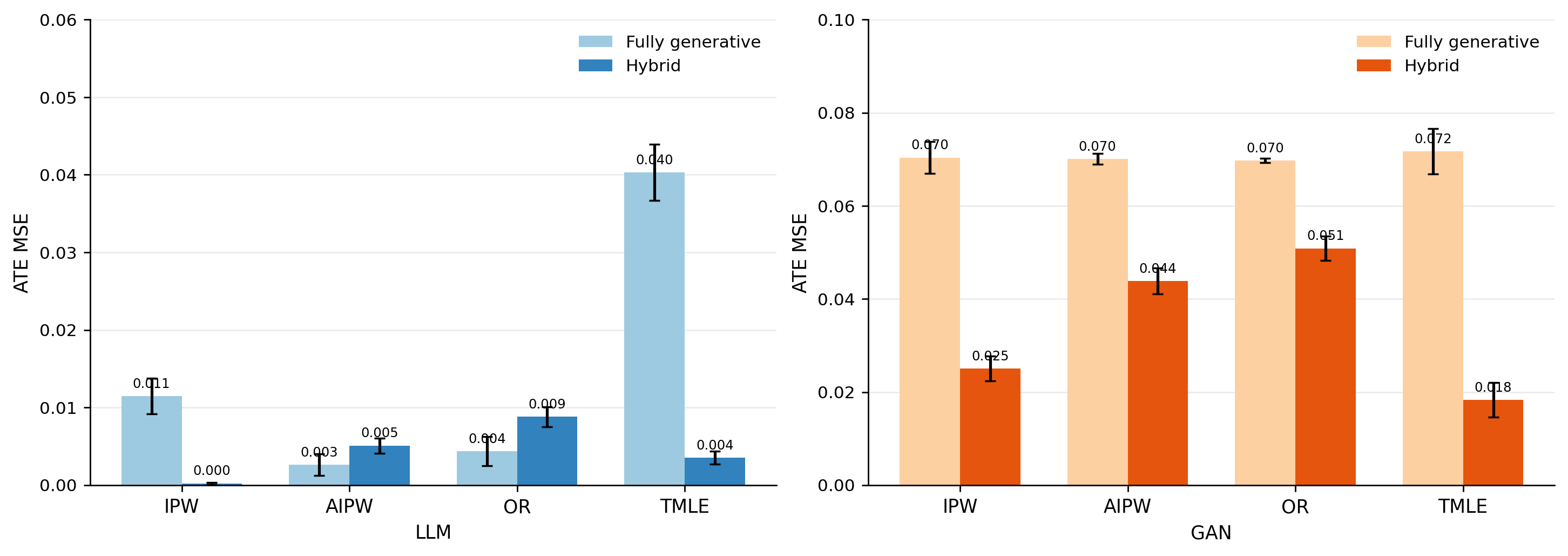}
    {\small $d=20$, $n_{\mathrm{seed}}=1000$, poor overlap, complex outcome}
\end{minipage}
\hfill
\begin{minipage}{0.48\textwidth}
    \centering
    \includegraphics[width=\textwidth]{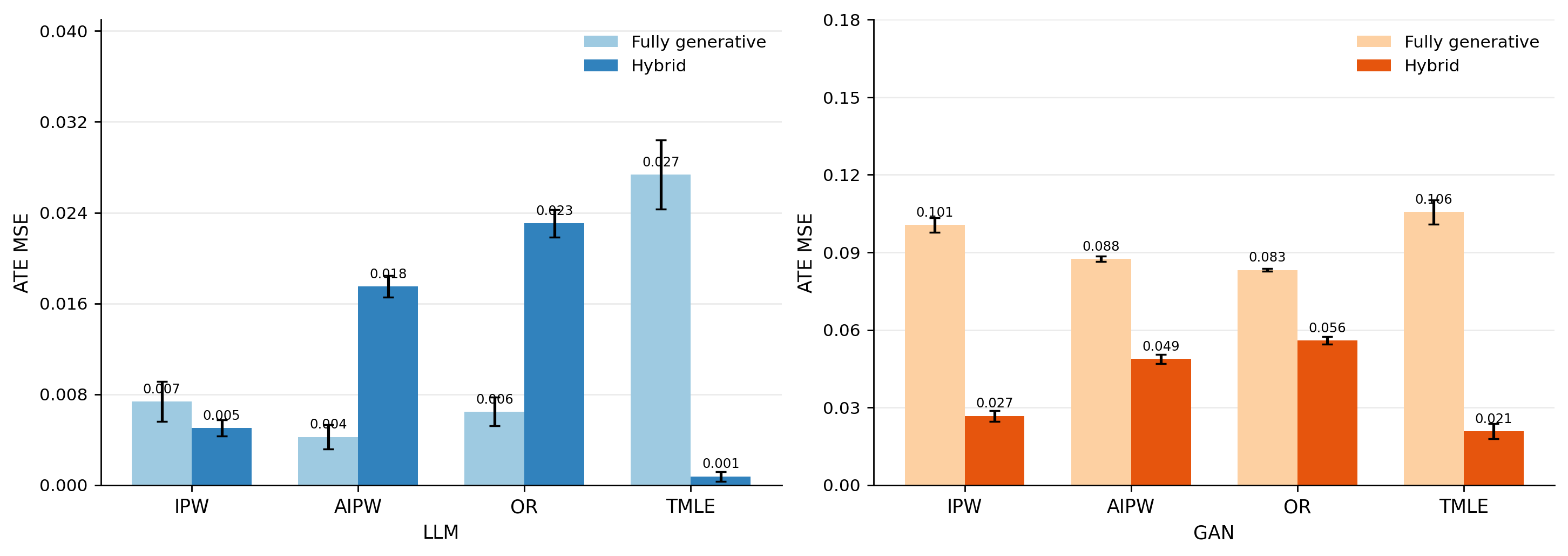}
    {\small $d=20$, $n_{\mathrm{seed}}=1000$, poor overlap, simple outcome}
\end{minipage}

\vspace{0.6em}

\begin{minipage}{0.48\textwidth}
    \centering
    \includegraphics[width=\textwidth]{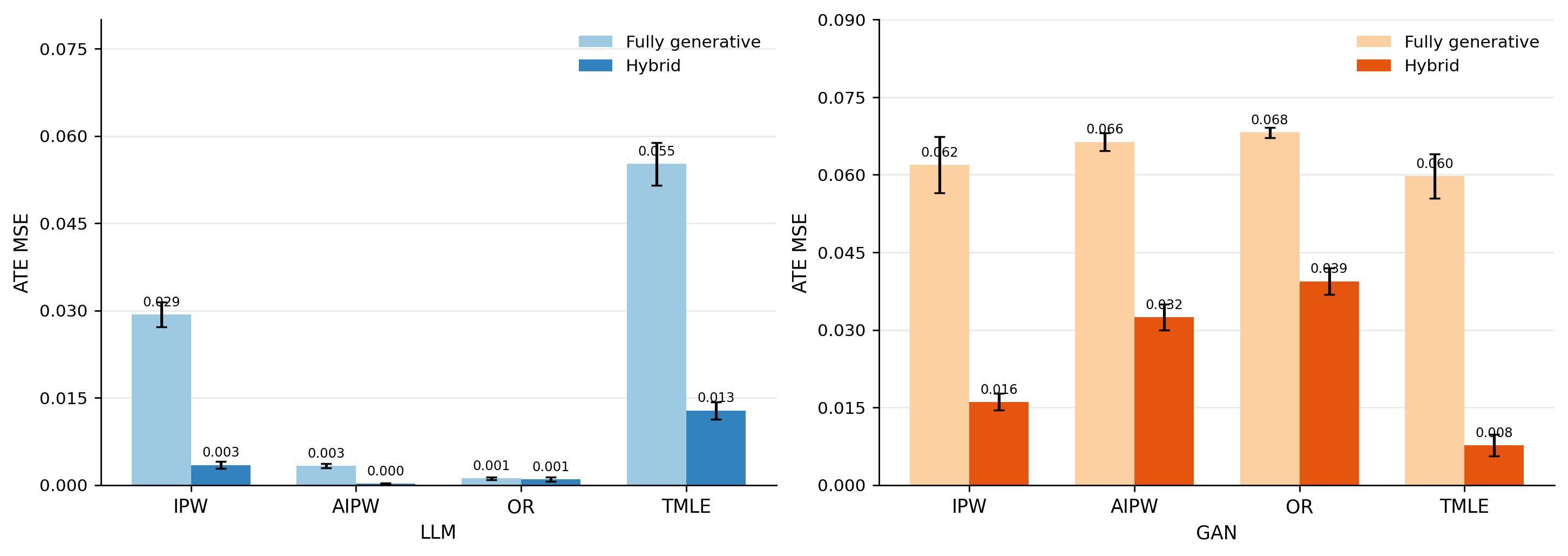}
    {\small $d=20$, $n_{\mathrm{seed}}=500$, poor overlap, complex outcome}
\end{minipage}
\hfill
\begin{minipage}{0.48\textwidth}
    \centering
    \includegraphics[width=\textwidth]{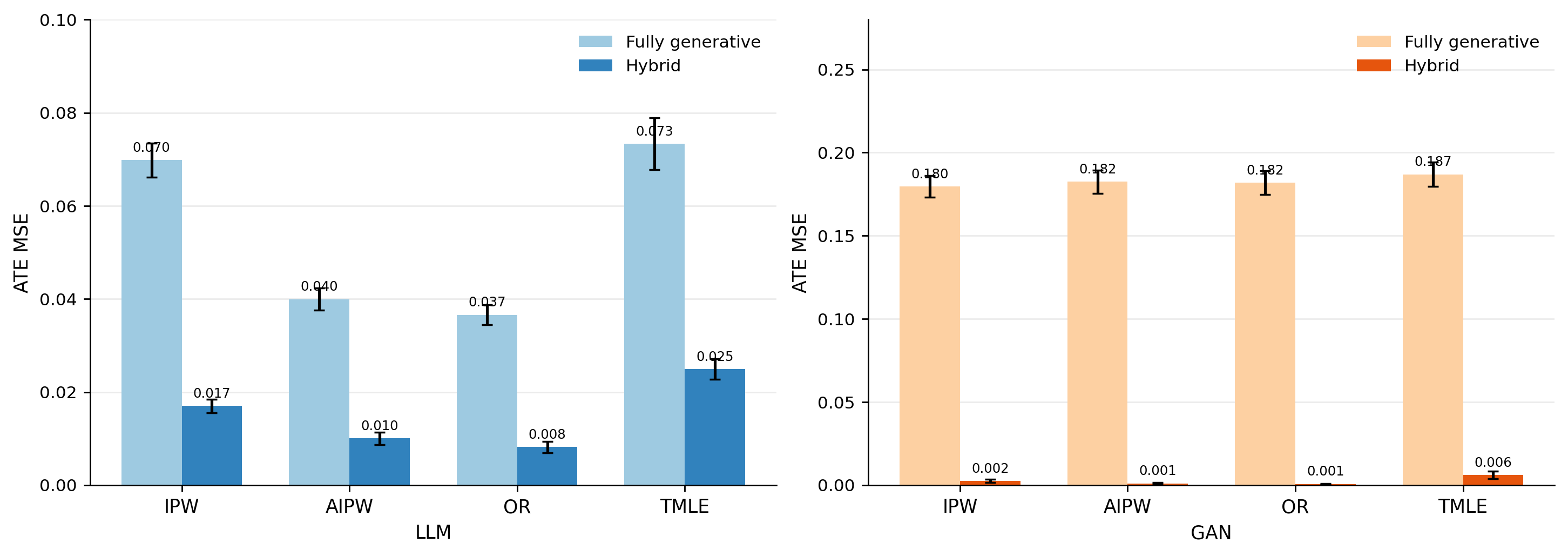}
    {\small $d=6$, $n_{\mathrm{seed}}=1000$, poor overlap, complex outcome}
\end{minipage}

\vspace{0.6em}

\begin{minipage}{0.55\textwidth}
    \centering
    \includegraphics[width=\textwidth]{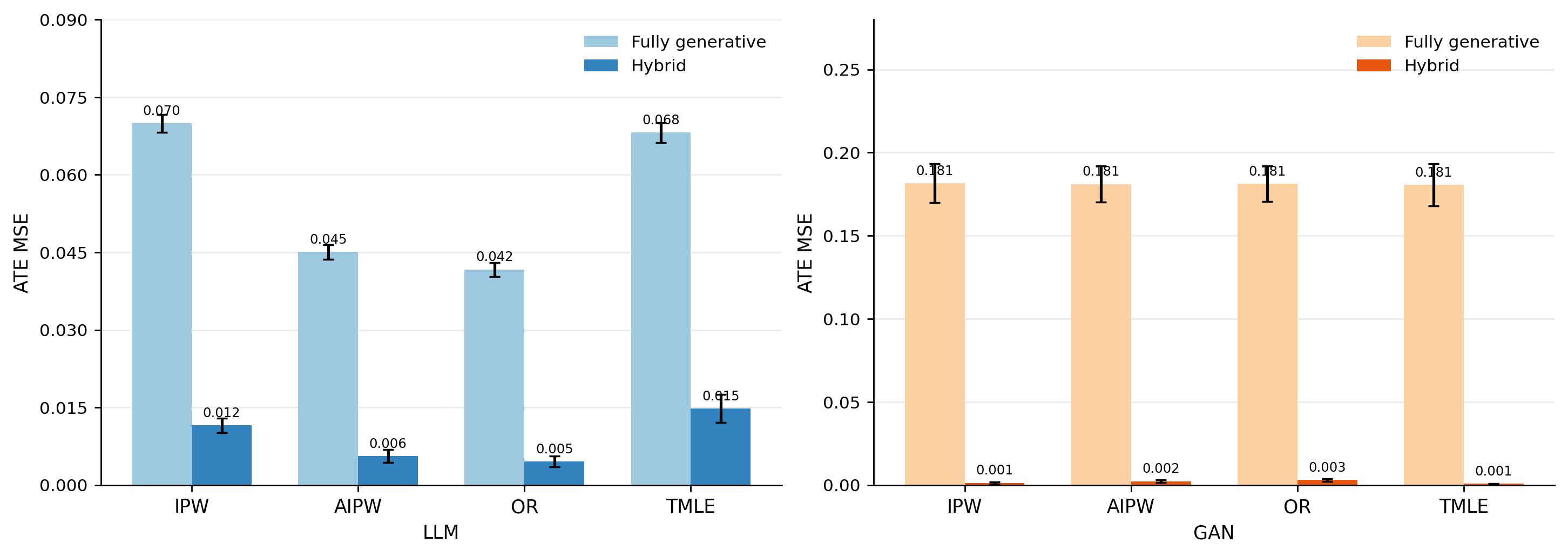}
    {\small $d=6$, $n_{\mathrm{seed}}=1000$, moderate overlap, complex outcome}
\end{minipage}

\caption{Additional ATE MSE stress tests using random-forest nuisance learners. Lower is better. Results vary covariate dimension, seed sample size, overlap strength, and outcome complexity. Hybrid synthesis gives the largest gains in complex regimes, while fully generative LLM synthesis can be competitive in the simpler setting.}
\label{fig:appendix_rf_outcome_mse}
\end{figure*}

\begin{figure*}[t]
\centering

\begin{minipage}{0.48\textwidth}
    \centering
    \includegraphics[width=\textwidth]{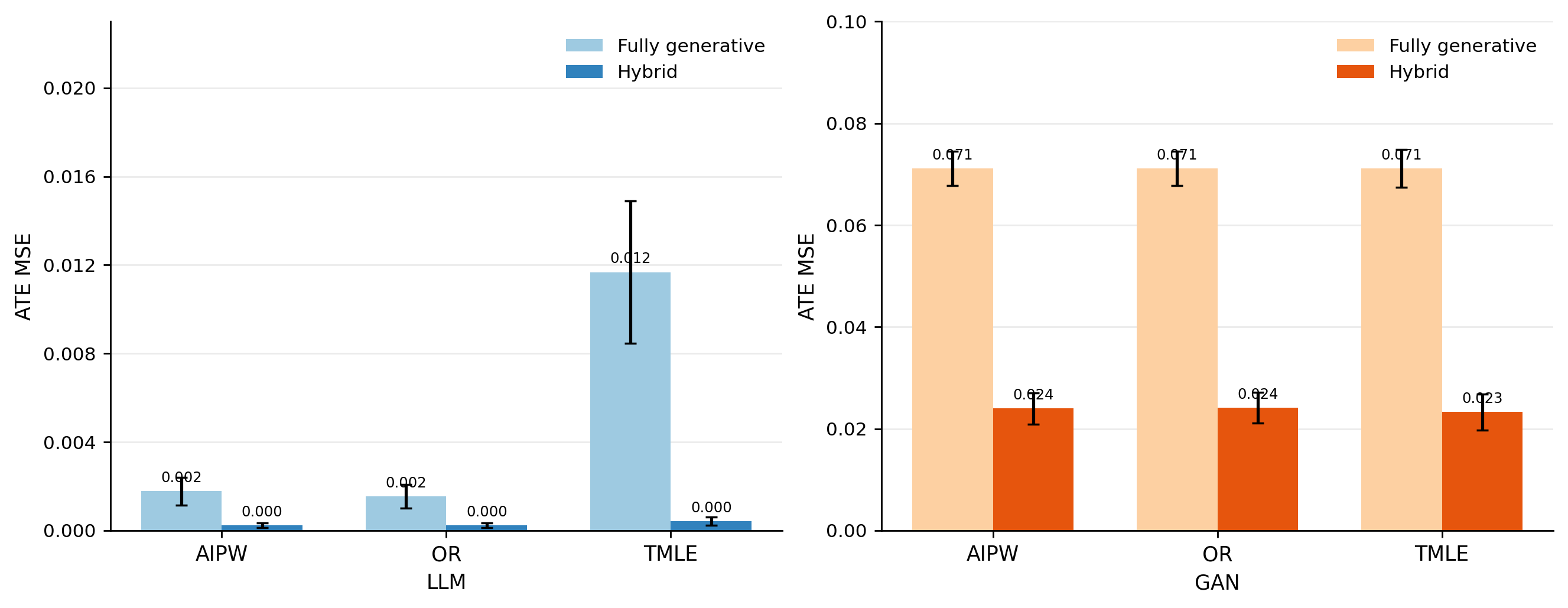}
    {\small $d=20$, $n_{\mathrm{seed}}=1000$, poor overlap, complex outcome}
\end{minipage}
\hfill
\begin{minipage}{0.48\textwidth}
    \centering
    \includegraphics[width=\textwidth]{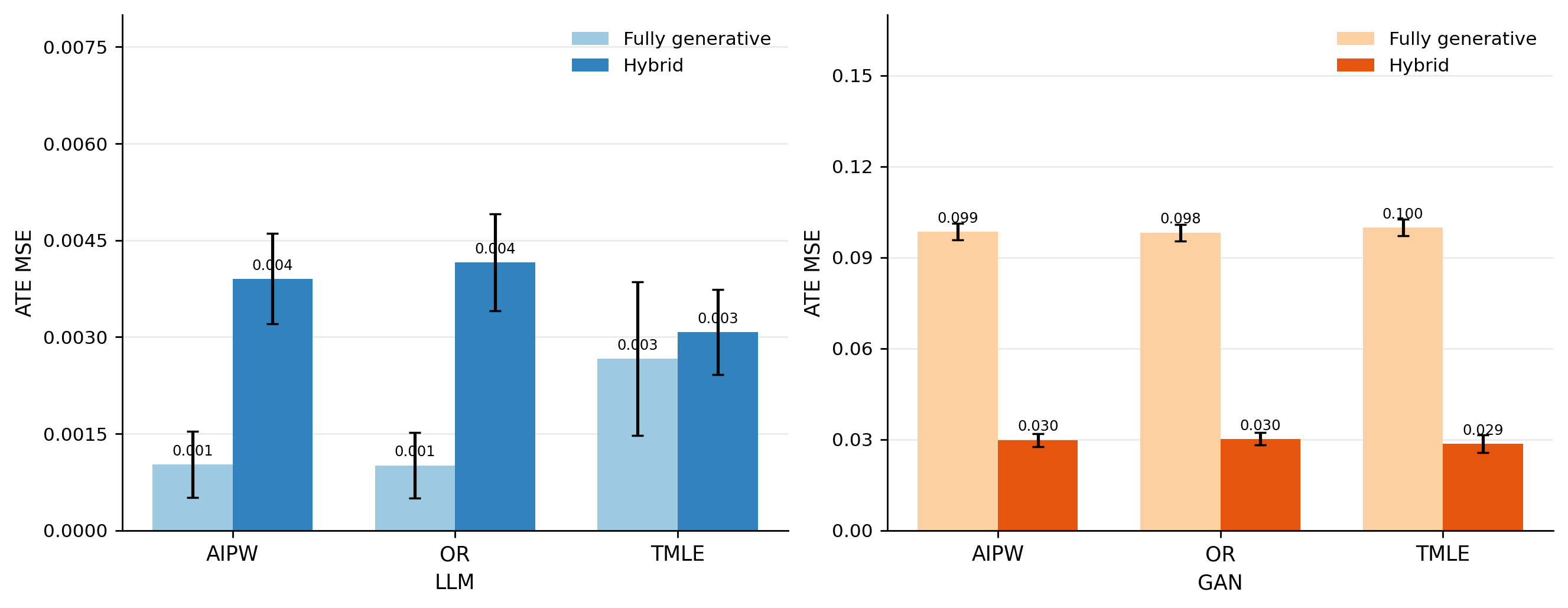}
    {\small $d=20$, $n_{\mathrm{seed}}=1000$, poor overlap, simple outcome}
\end{minipage}

\vspace{0.6em}

\begin{minipage}{0.48\textwidth}
    \centering
    \includegraphics[width=\textwidth]{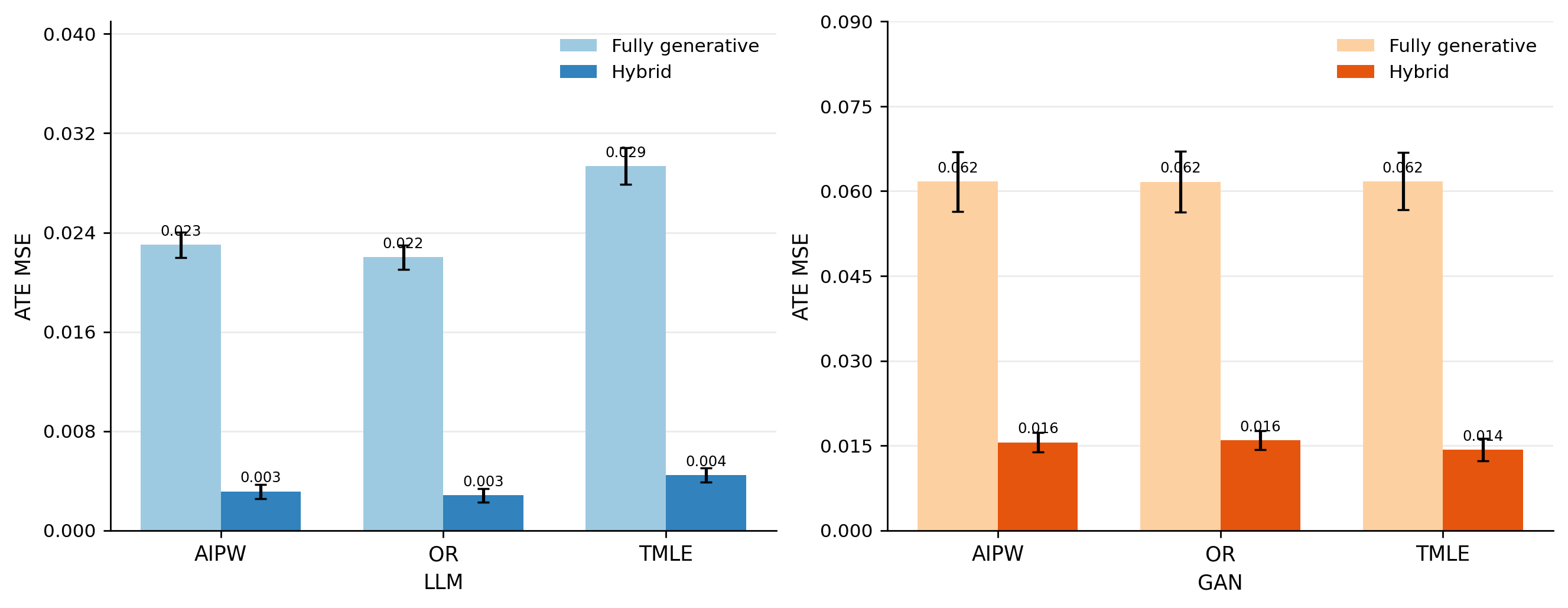}
    {\small $d=20$, $n_{\mathrm{seed}}=500$, poor overlap, complex outcome}
\end{minipage}
\hfill
\begin{minipage}{0.48\textwidth}
    \centering
    \includegraphics[width=\textwidth]{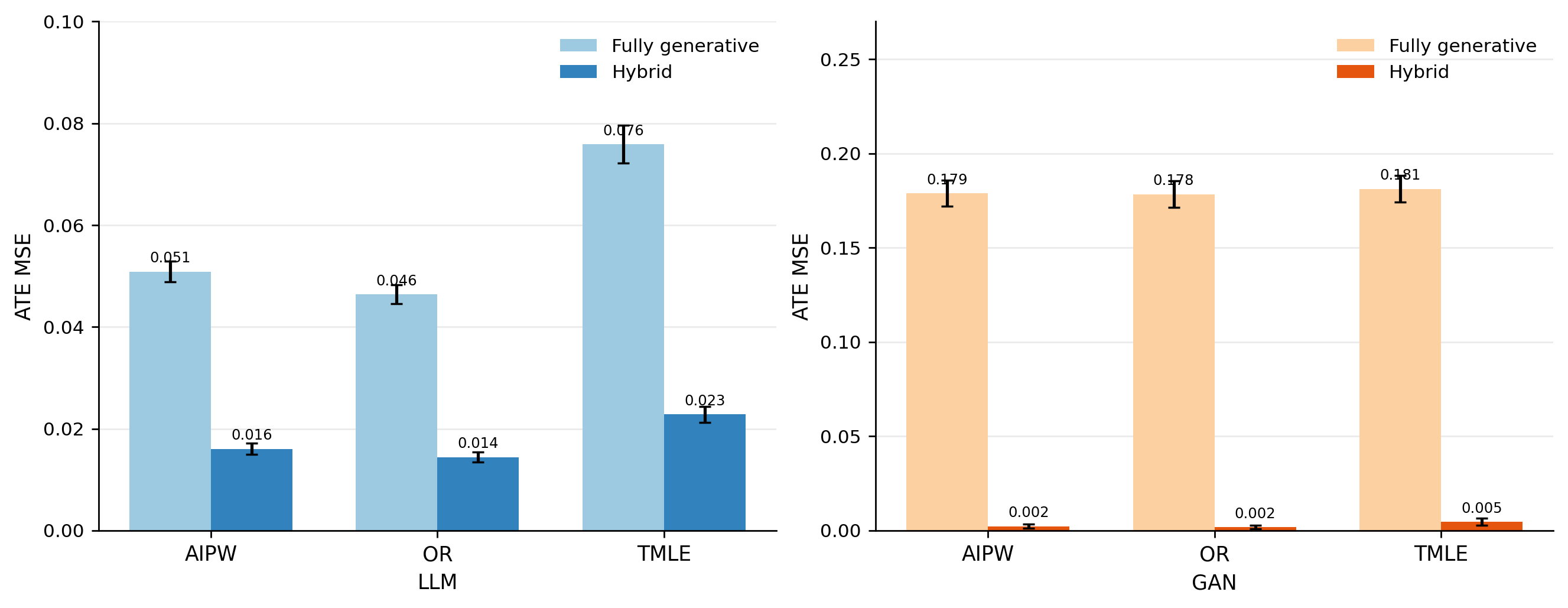}
    {\small $d=6$, $n_{\mathrm{seed}}=1000$, poor overlap, complex outcome}
\end{minipage}

\vspace{0.6em}

\begin{minipage}{0.55\textwidth}
    \centering
    \includegraphics[width=\textwidth]{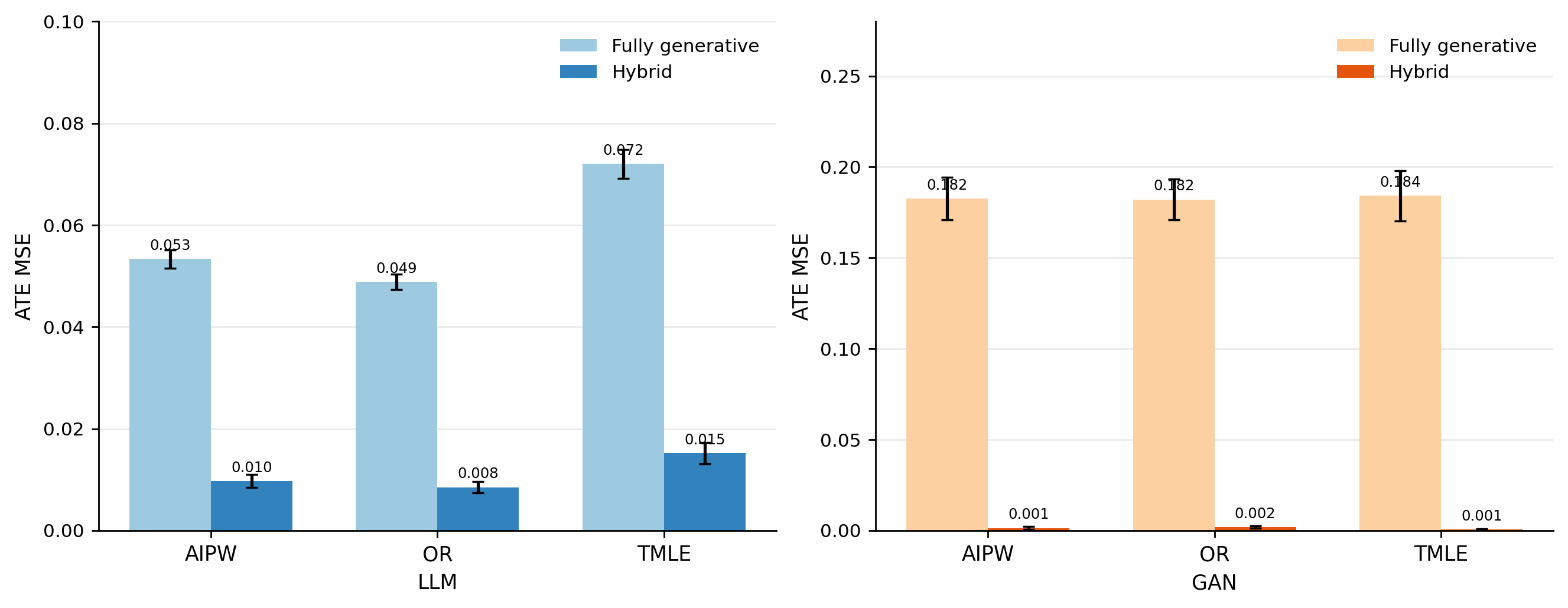}
    {\small $d=6$, $n_{\mathrm{seed}}=1000$, moderate overlap, complex outcome}
\end{minipage}

\caption{Additional ATE MSE stress tests using logistic outcome learners for outcome regression, AIPW, and TMLE. IPW is omitted because it does not use an outcome model. Lower is better. The qualitative pattern remains similar to the main random-forest comparison: hybrid synthesis is most beneficial in complex or limited-overlap settings, while fully generative LLM synthesis remains competitive in the simpler outcome-model setting.}
\label{fig:appendix_logistic_outcome_mse}
\end{figure*}

\begin{figure*}[t]
\centering

\begin{minipage}{0.48\textwidth}
    \centering
    \includegraphics[width=\textwidth]{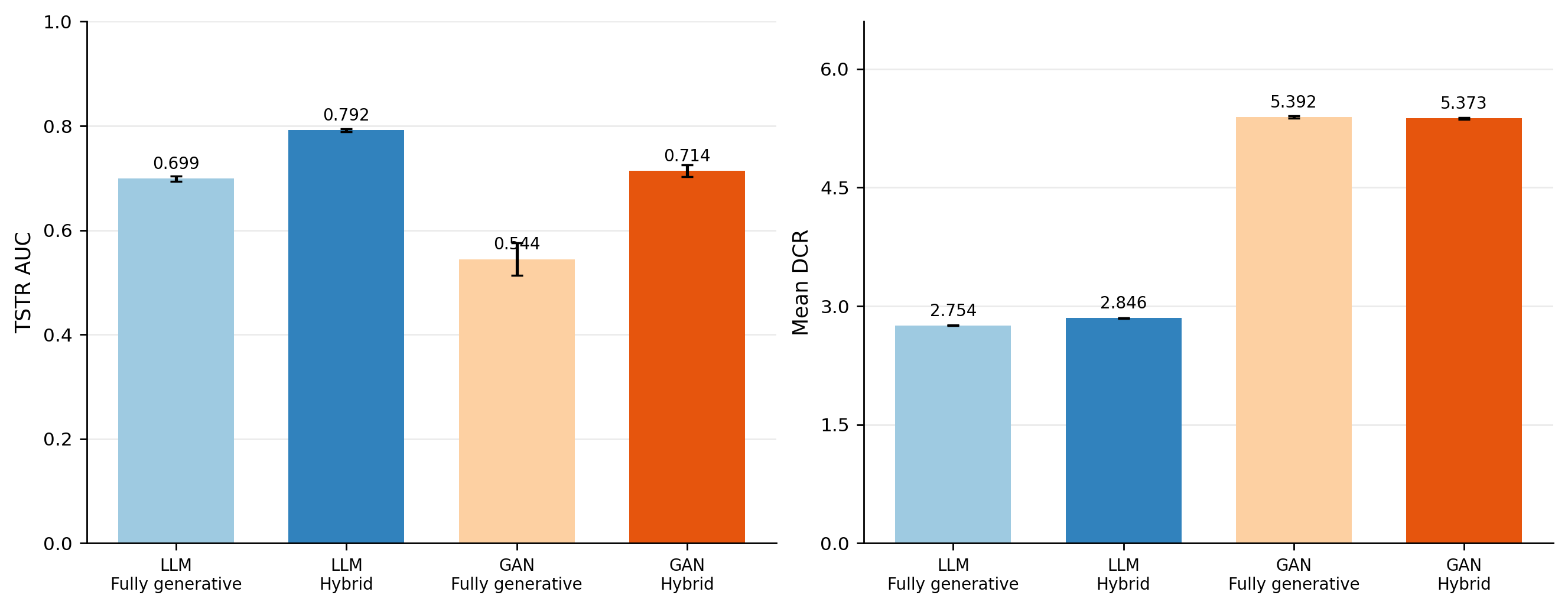}
    {\small $d=20$, $n_{\mathrm{seed}}=1000$, poor overlap, complex}
\end{minipage}
\hfill
\begin{minipage}{0.48\textwidth}
    \centering
    \includegraphics[width=\textwidth]{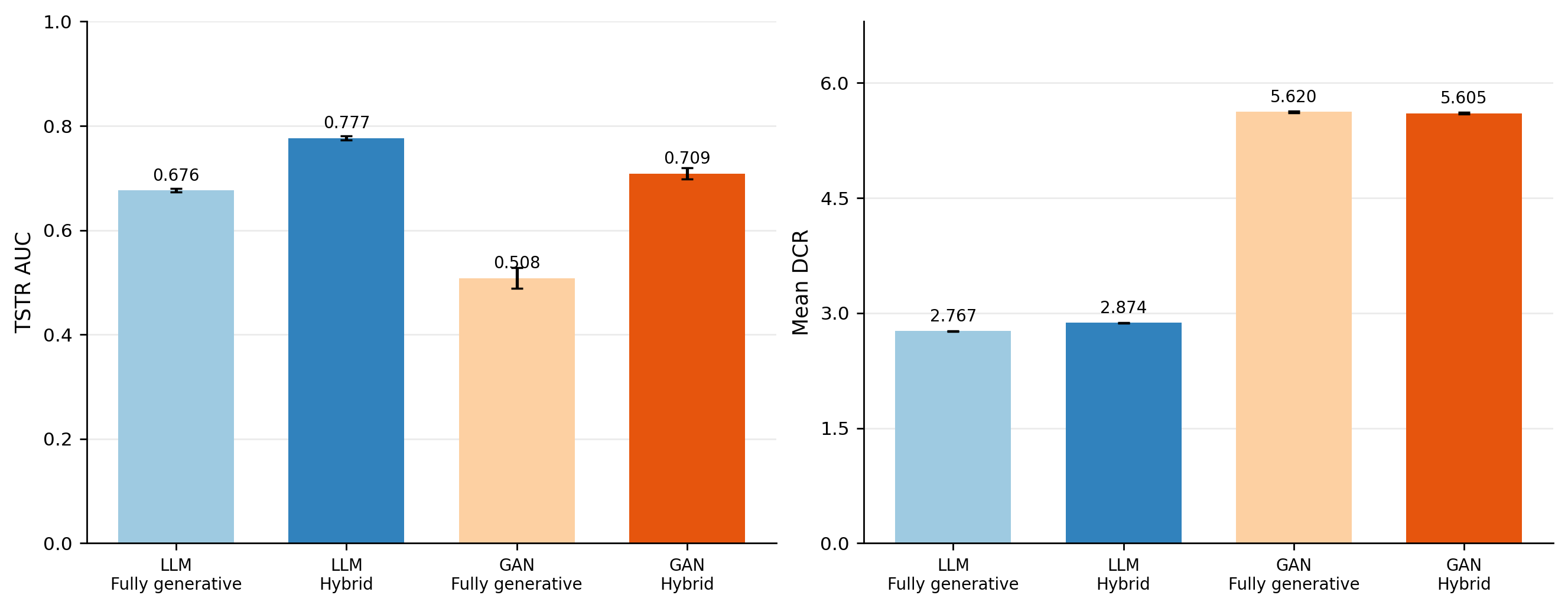}
    {\small $d=20$, $n_{\mathrm{seed}}=1000$, poor overlap, simple}
\end{minipage}

\vspace{0.6em}

\begin{minipage}{0.48\textwidth}
    \centering
    \includegraphics[width=\textwidth]{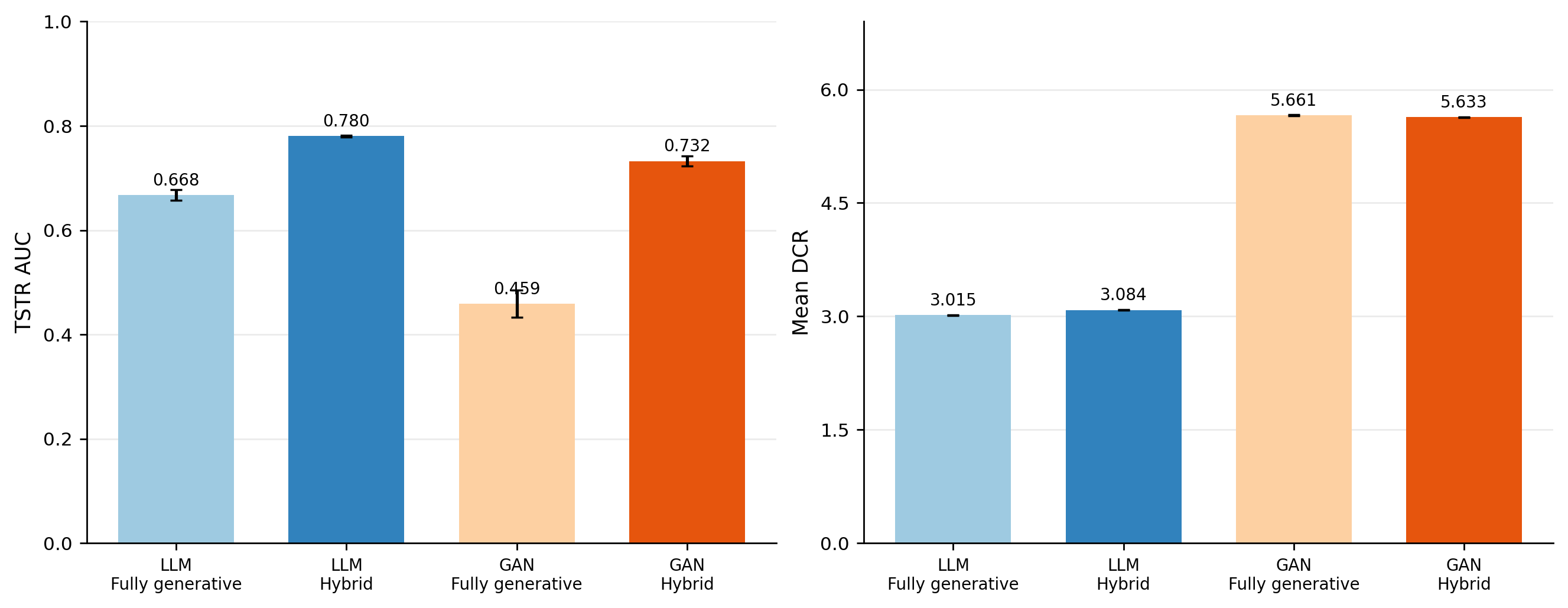}
    {\small $d=20$, $n_{\mathrm{seed}}=500$, poor overlap, complex}
\end{minipage}
\hfill
\begin{minipage}{0.48\textwidth}
    \centering
    \includegraphics[width=\textwidth]{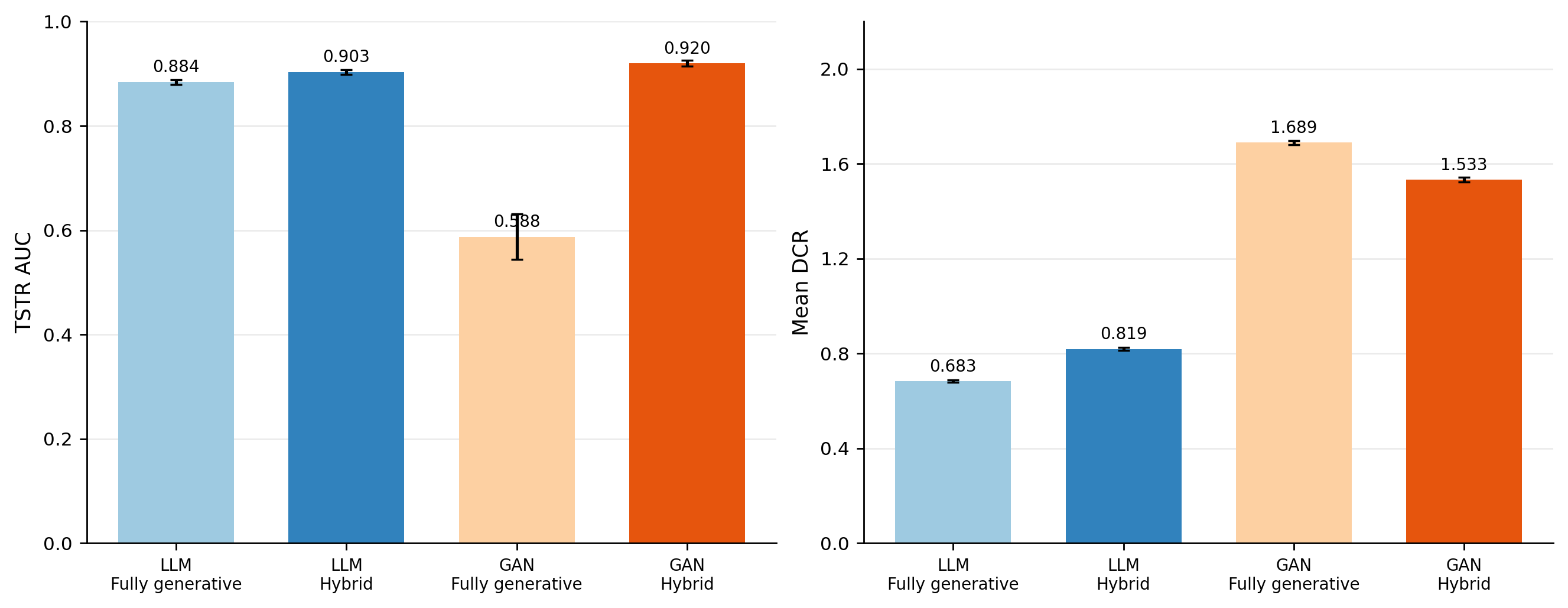}
    {\small $d=6$, $n_{\mathrm{seed}}=1000$, poor overlap, complex}
\end{minipage}

\vspace{0.6em}

\begin{minipage}{0.55\textwidth}
    \centering
    \includegraphics[width=\textwidth]{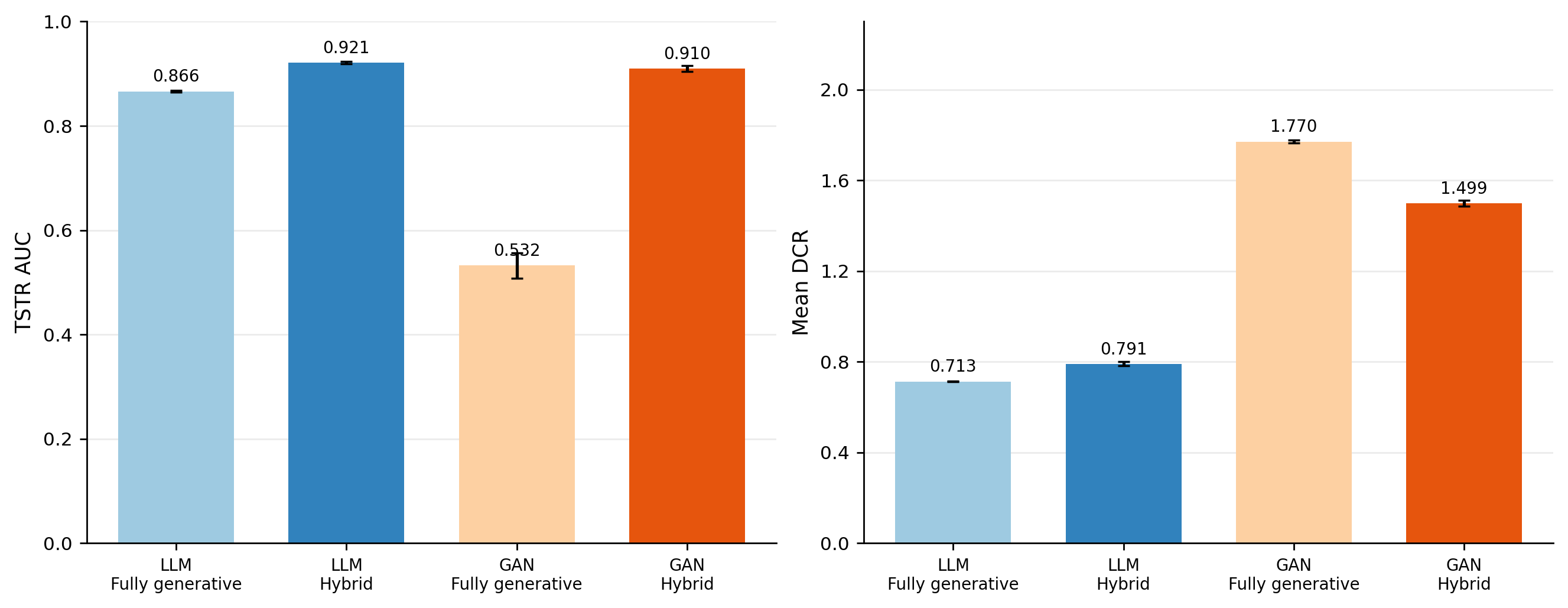}
    {\small $d=6$, $n_{\mathrm{seed}}=1000$, moderate overlap, complex}
\end{minipage}

\caption{TSTR and DCR diagnostics across simulation settings.}
\label{fig:appendix_tstr_dcr_vary}
\end{figure*}

\begin{figure*}[t]
    \centering
    \begin{minipage}{0.48\textwidth}
        \centering
        \includegraphics[width=\textwidth]{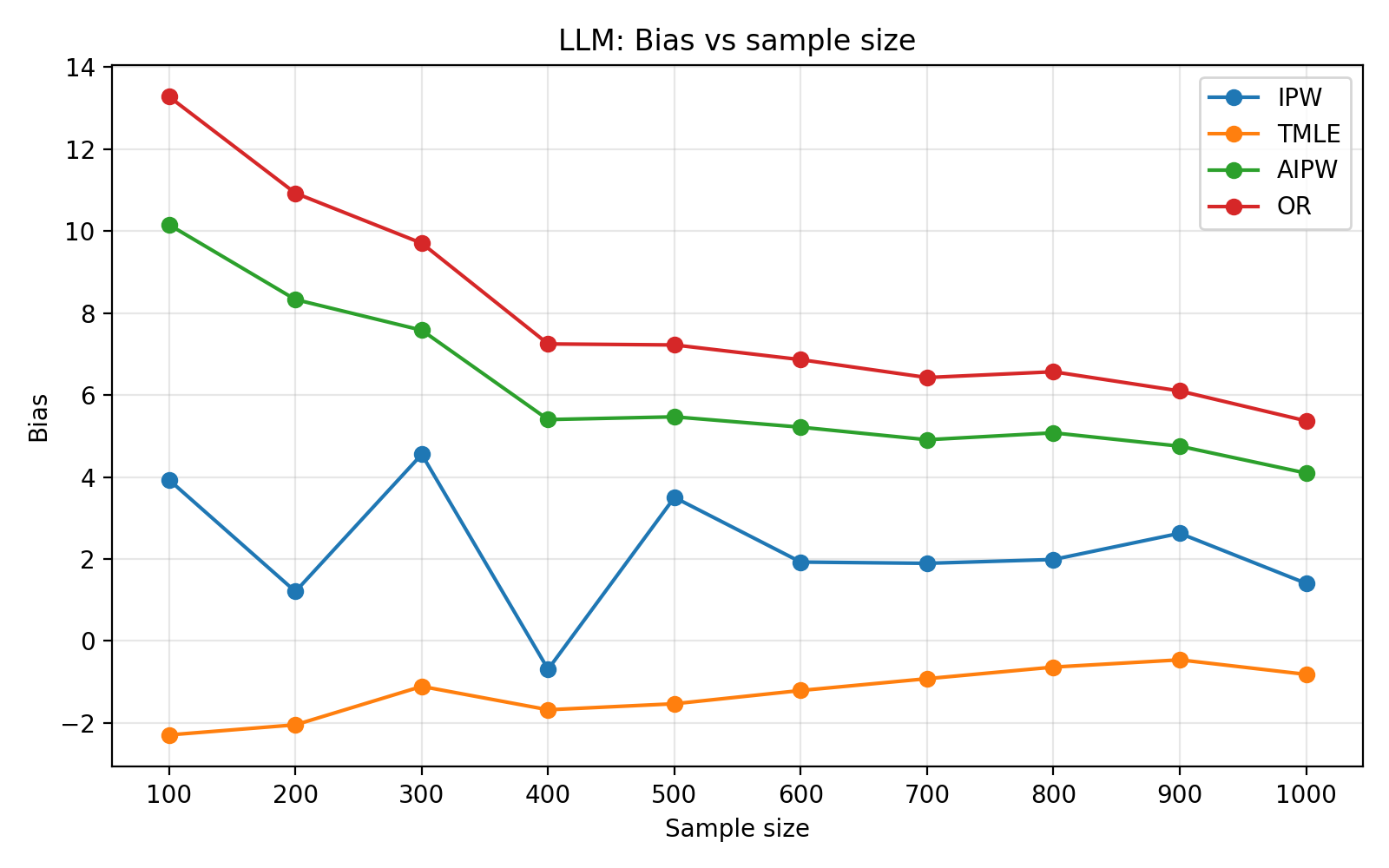}
    \end{minipage}
    \hfill
    \begin{minipage}{0.48\textwidth}
        \centering
        \includegraphics[width=\textwidth]{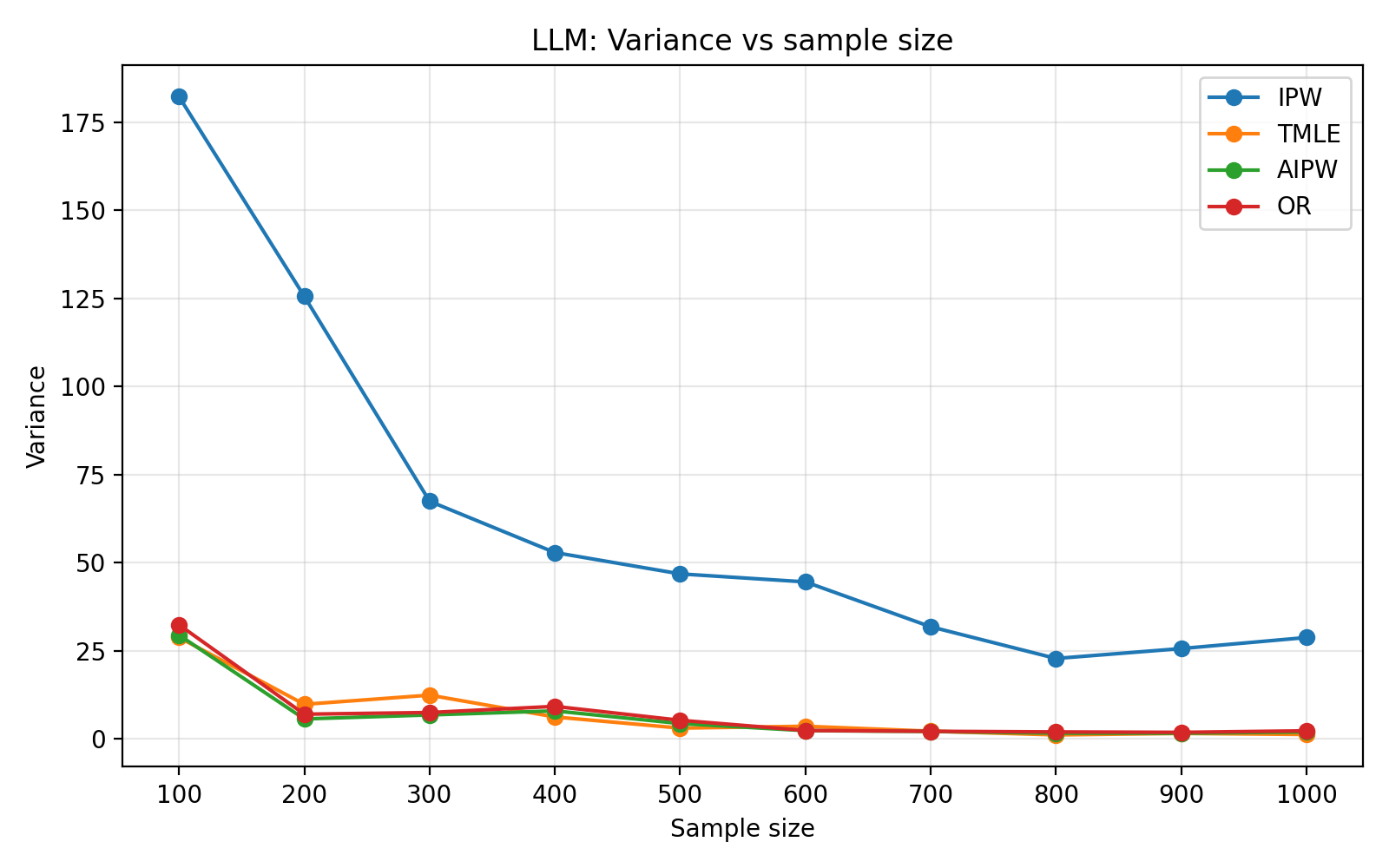}
    \end{minipage}

    \vspace{0.6em}

    \begin{minipage}{0.48\textwidth}
        \centering
        \includegraphics[width=\textwidth]{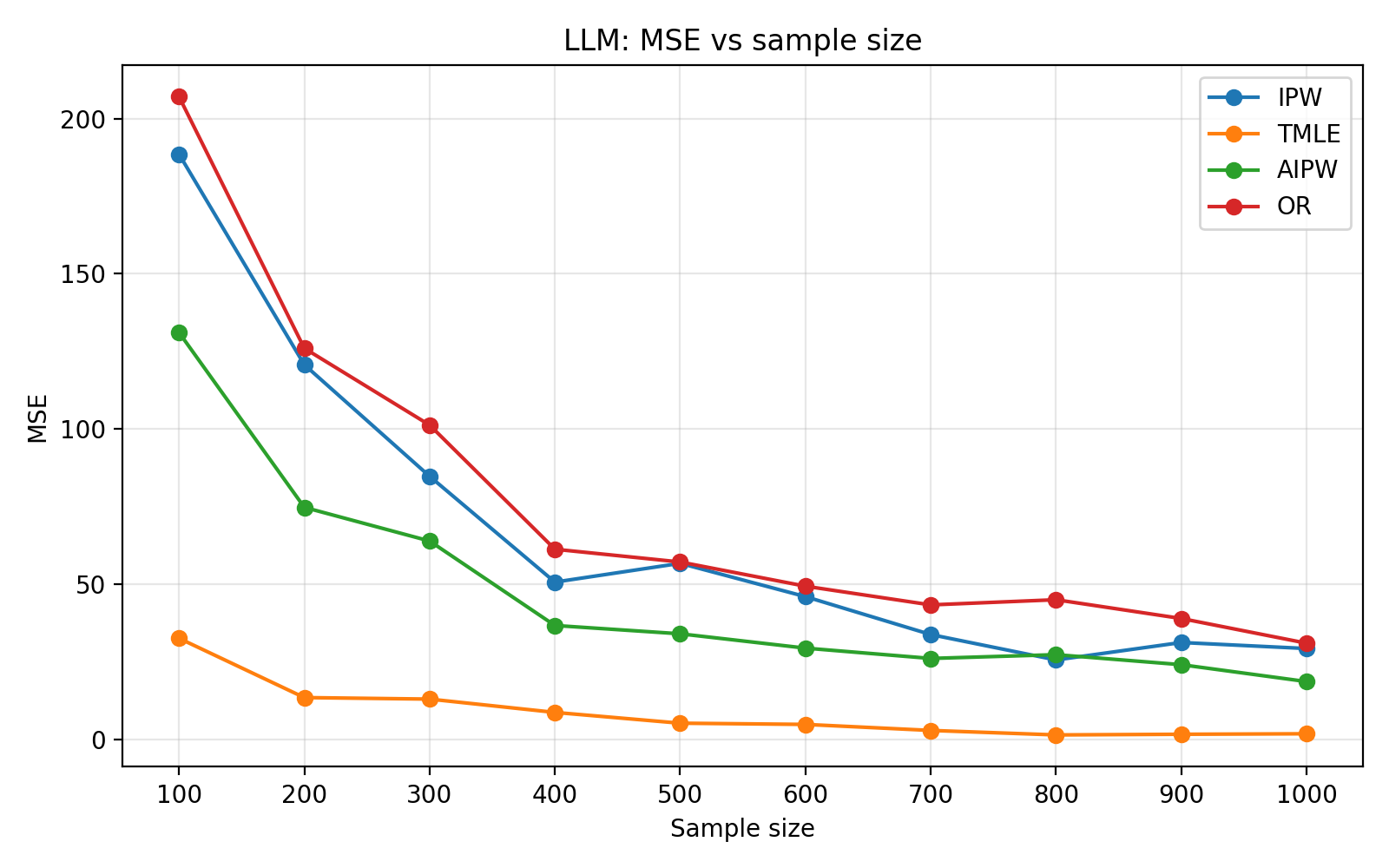}
    \end{minipage}
    \hfill
    \begin{minipage}{0.48\textwidth}
        \centering
        \includegraphics[width=\textwidth]{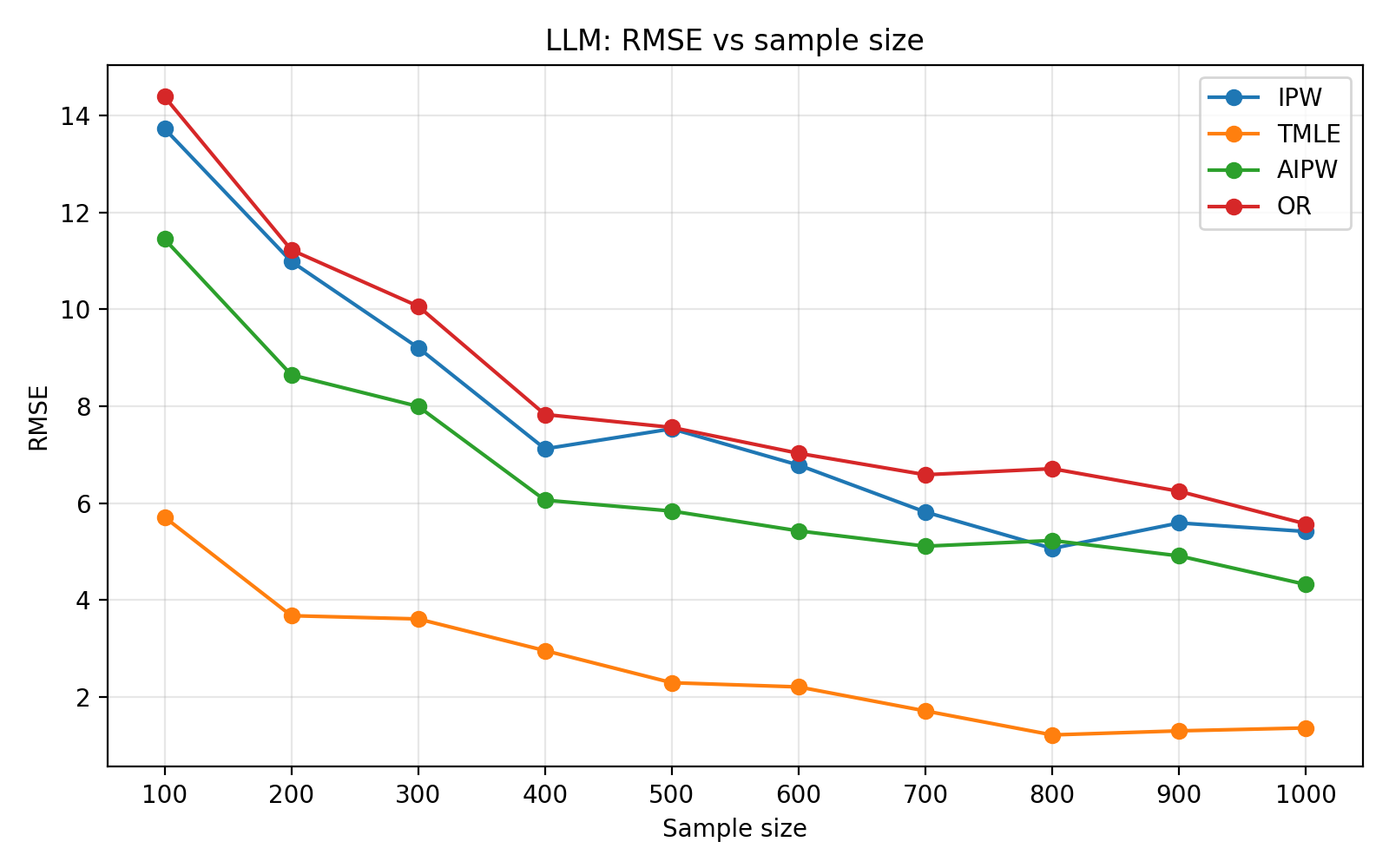}
    \end{minipage}
    \caption{Finite-sample benchmarking results under the LLM-based hybrid ACTG simulator. Panels report bias, variance, MSE, and RMSE as functions of sample size for IPW, AIPW, outcome regression, and TMLE.}
    \label{fig:actg_llm_2x2}
\end{figure*}

\begin{figure*}[t]
    \centering
    \begin{minipage}{0.48\textwidth}
        \centering
        \includegraphics[width=\textwidth]{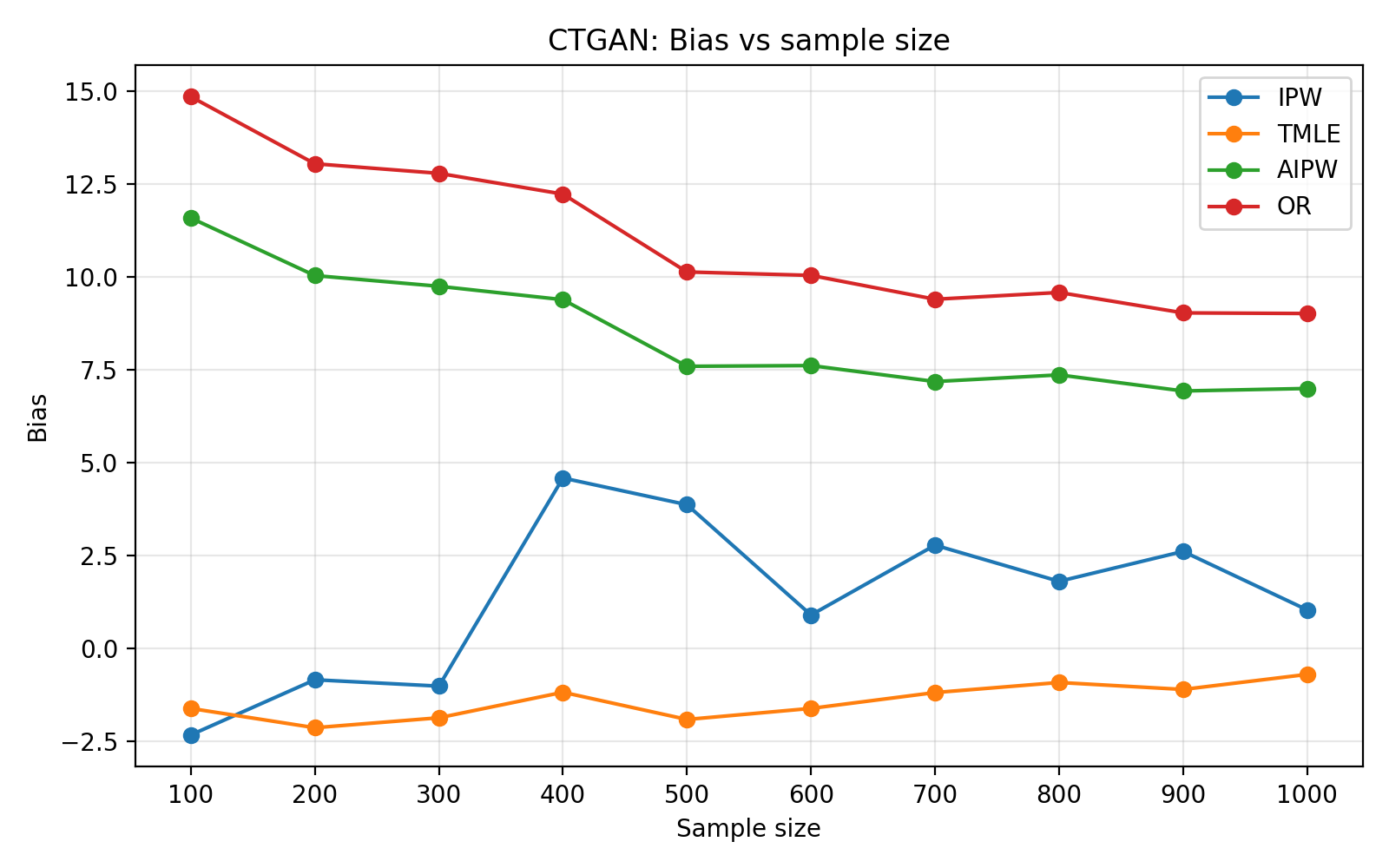}
    \end{minipage}
    \hfill
    \begin{minipage}{0.48\textwidth}
        \centering
        \includegraphics[width=\textwidth]{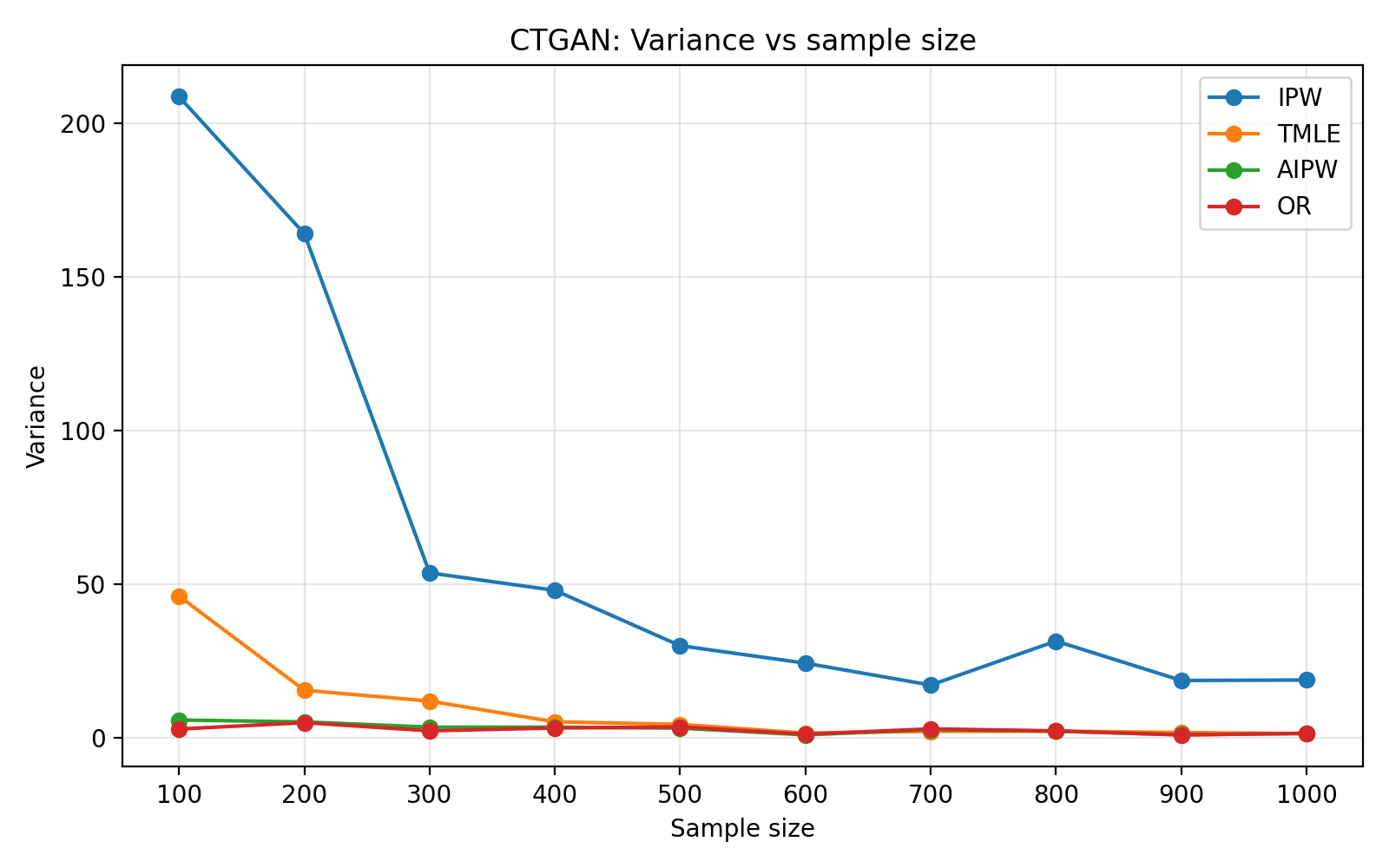}
    \end{minipage}

    \vspace{0.6em}

    \begin{minipage}{0.48\textwidth}
        \centering
        \includegraphics[width=\textwidth]{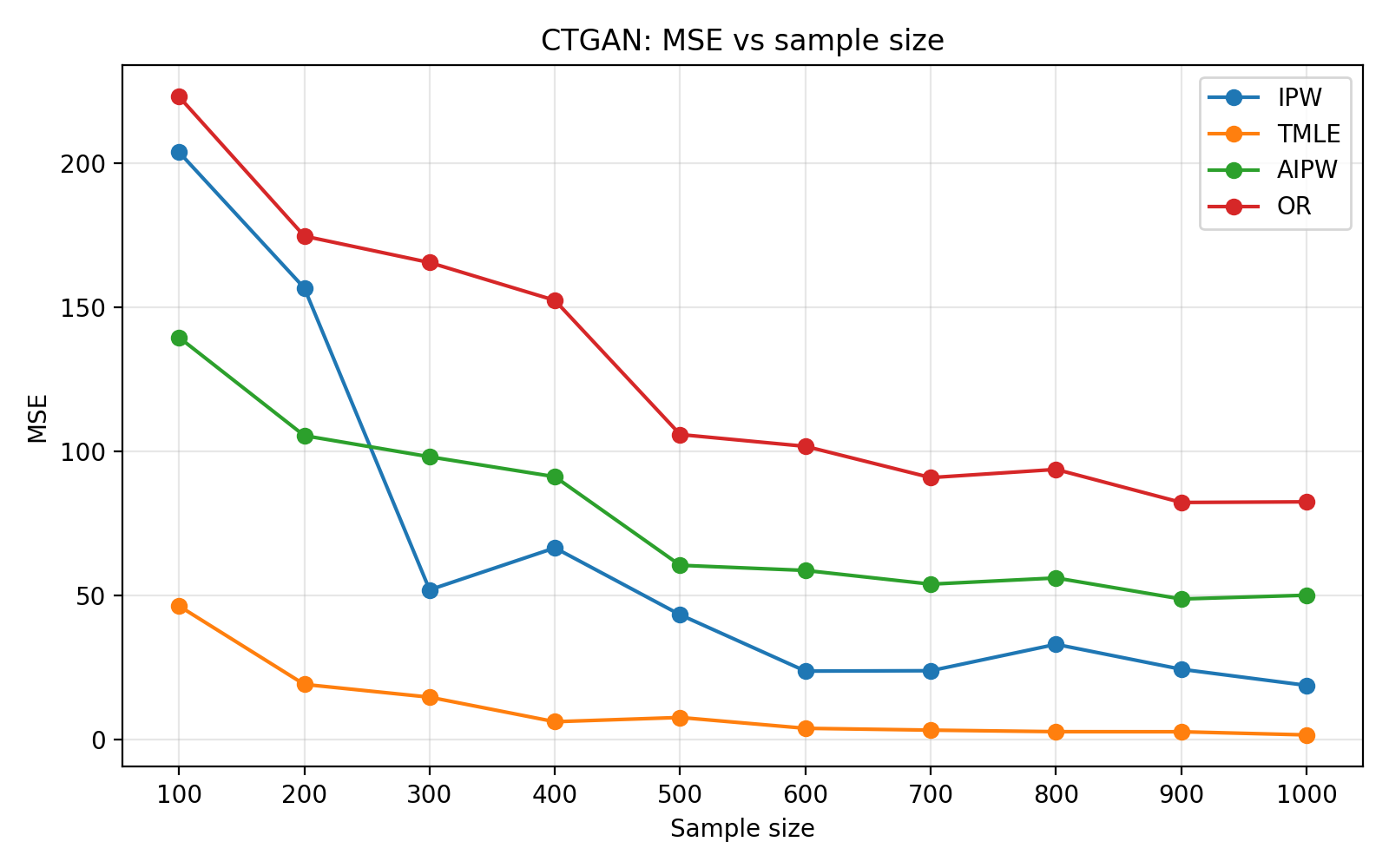}
    \end{minipage}
    \hfill
    \begin{minipage}{0.48\textwidth}
        \centering
        \includegraphics[width=\textwidth]{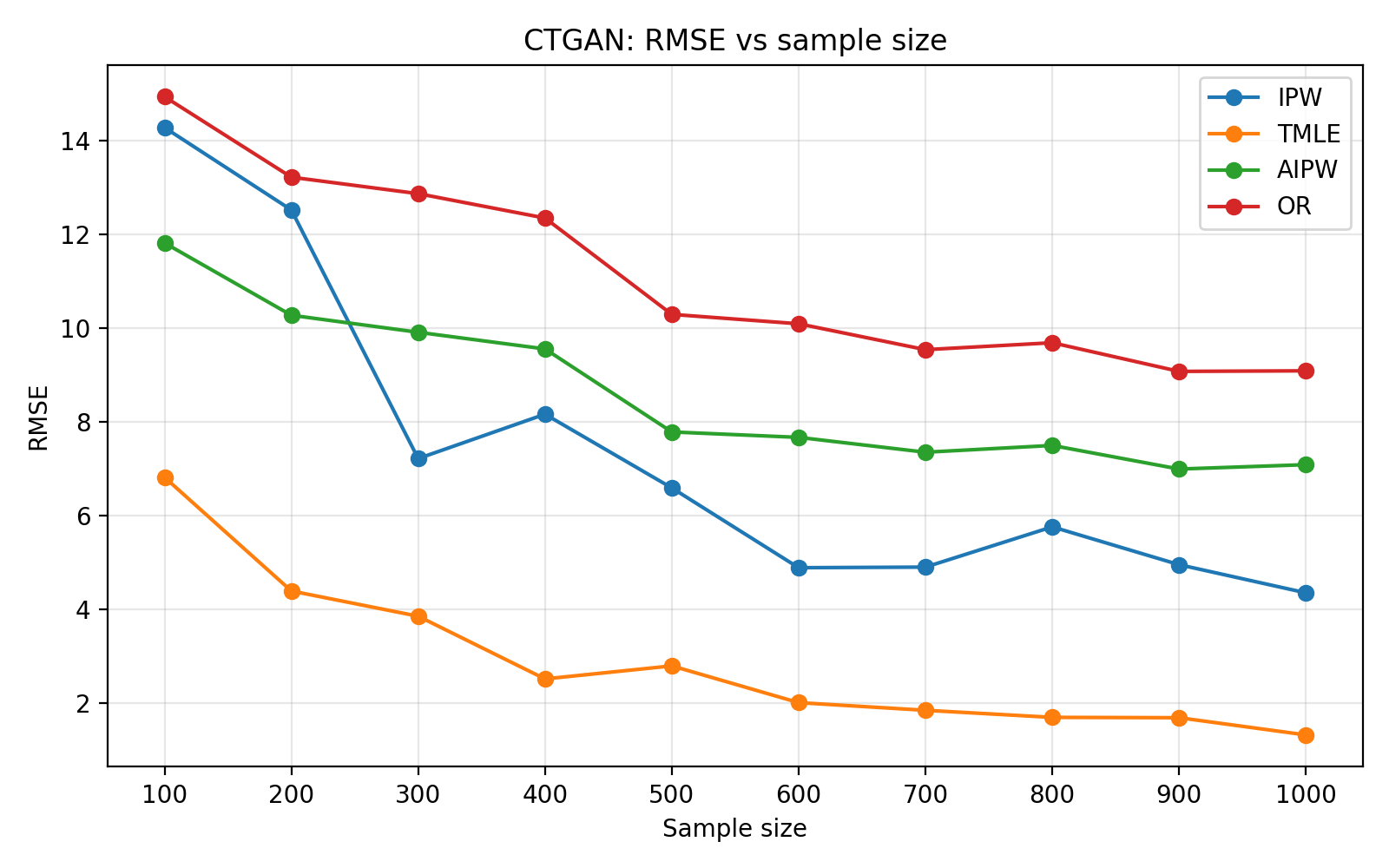}
    \end{minipage}
    \caption{Finite-sample benchmarking results under the CTGAN-based hybrid ACTG simulator. Panels report bias, variance, MSE, and RMSE as functions of sample size for IPW, AIPW, outcome regression, and TMLE.}
    \label{fig:actg_ctgan_2x2}
\end{figure*}

As an additional calibration check, we evaluated all four estimators on the full LLM hybrid synthetic pool before studying finite-sample subsamples. The full LLM synthetic pool contains $49{,}749$ observations. The large-sample TMLE estimate is $-23.5240$, while the corresponding large-sample estimates for IPW, AIPW, and outcome regression are $-21.5921$, $-22.8787$, and $-22.7074$, respectively. Thus, AIPW and TMLE are close on the full synthetic pool, differing by only $0.6453$, and outcome regression is also within $0.8166$ of TMLE. IPW differs more from TMLE, by $1.9319$; this is not unexpected because IPW is a weight-based estimator and a single large synthetic pool can still exhibit non-negligible fluctuation due to high-variance weights. This calibration helps interpret the finite-sample simulation results. Since AIPW and TMLE are close on the full synthetic pool, the large positive AIPW bias observed at sample size $n=1000$ is unlikely to be explained merely by AIPW targeting a substantially different large-sample estimand. Instead, it reflects finite-sample instability of AIPW in this simulated ACTG environment. At $n=1000$, the AIPW mean remains $4.0963$ units above the full-pool TMLE reference and $3.4509$ units above its own full-pool AIPW estimate. Outcome regression shows an even larger finite-sample discrepancy, remaining $5.3706$ units above the full-pool TMLE reference and $4.5539$ units above its own full-pool estimate. By contrast, TMLE is much closer to its large-sample reference at $n=1000$, with bias $-0.8140$ and MSE $1.8414$. IPW is comparatively close to its own full-pool estimate at $n=1000$, but its MSE remains large, consistent with the high-variance behavior of weighting estimators. These results support the role of the synthetic simulator as a finite-sample diagnostic: estimators that agree in the large synthetic population can still behave very differently at realistic sample sizes.

\subsection*{Reproducibility}

We describe the simulation design and implementation details. The synthetic-data experiments use binary treatment \(A\), binary outcome \(Y\), and covariates \(W=(W_1,\ldots,W_d)\). We consider six settings: the primary main-text setting and five additional stress-test settings.

\[
\begin{array}{ll}
\text{(i)}   & d=6,\ n_{\mathrm{seed}}=1000,\ \text{randomized seed treatment, complex treatment effect},\\
\text{(ii)}  & d=6,\ n_{\mathrm{seed}}=1000,\ \text{moderate overlap, complex treatment effect},\\
\text{(iii)} & d=6,\ n_{\mathrm{seed}}=1000,\ \text{poor overlap, complex treatment effect},\\
\text{(iv)}  & d=20,\ n_{\mathrm{seed}}=1000,\ \text{poor overlap, complex treatment effect},\\
\text{(v)}   & d=20,\ n_{\mathrm{seed}}=500,\ \text{poor overlap, complex treatment effect},\\
\text{(vi)}  & d=20,\ n_{\mathrm{seed}}=1000,\ \text{poor overlap, simple treatment effect}.
\end{array}
\]

In each setting, the reference population has size \(50{,}000\), and the test set used for TSTR evaluation has size \(1000\). The observational datasets in the positivity experiments have sizes \(100\), \(200\), and \(500\), while the ATE preservation experiments use synthetic datasets of size \(1000\). The true ATE is computed by averaging the difference between treated and control outcome probabilities over the large reference population. The \(d=6\) settings include three binary and three continuous covariates, with \(W_3\) depending on \(W_1\) and \(W_2\), and \(W_6\) depending on \(W_4\) and \(W_5\). The primary main-text setting uses randomized treatment assignment and a nonlinear heterogeneous treatment effect. The five stress-test settings vary overlap, dimension, seed sample size, and treatment-effect complexity. The \(d=20\) settings use a nonlinear dependent covariate system with shared latent factors, correlated Gaussian noise, and nuisance covariates correlated with outcome-relevant variables.

For fully generative LLM synthesis, we use GReaT~\citep{borisov23} with GPT-2. The model is trained on the seed table containing covariates, treatment, and outcome for 50 epochs with batch size 32. Rows are serialized using the GReaT tabular-text representation with randomized feature order, and generated text is converted back to tabular format after sampling. For CTGAN synthesis, we train CTGAN~\citep{xu19} on the same seed table for 50 epochs, specifying binary variables as discrete columns and the remaining variables as continuous columns. For both generators, generated samples are postprocessed to enforce the expected column order, numeric types, and valid binary support.

Hybrid datasets use the corresponding LLM or CTGAN model only for covariate generation. Treatment and outcome are then generated using separately fitted nuisance mechanisms. In the randomized hybrid design, synthetic treatment is assigned independently with probability \(1/2\), and the outcome is generated from the fitted outcome mechanism. Thus, fully generative datasets model the full row jointly, whereas hybrid datasets model the covariate law, treatment mechanism, and outcome mechanism separately.

For estimator evaluation, we consider outcome regression, IPW, AIPW, and TMLE. The main specification uses random-forest nuisance learners for both propensity and outcome models where applicable. As a parametric outcome-model misspecification check, we also evaluate a logistic-outcome specification: outcome regression, AIPW, and TMLE use logistic regression for the outcome model, while propensity-score models remain random forests. IPW does not use an outcome model, so it is omitted from this misspecification plot and retained only in the main estimator comparison as a weighting-only benchmark. Extreme propensity estimates are truncated for numerical stability. Reported ATE MSE values are computed over five repetitions, with standard errors computed across repetitions.

DCR is used as a privacy-distance diagnostic. TSTR and ATE MSE are reported as separate predictive and causal diagnostics. For ACTG, the outcome \texttt{cd420} is continuous, so TSTR is reported as RMSE. A random-forest regressor is trained on each synthetic dataset using covariates and treatment as predictors and evaluated on the real ACTG data. The ACTG experiment is treated as an illustrative diagnostic study because the real-data ground-truth ATE is unknown.

\section*{Theoretical analysis}

\begin{proof}[Proof for \Cref{prop:ate_sensitivity}]
Write
\[
\Psi(P_W,Q)=\mathbb E_{P_W}[\Delta_Q(W)],
\qquad
\Psi(P_W^\star,Q^\star)=\mathbb E_{P_W^\star}[\Delta^\star(W)].
\]
Add and subtract \(\mathbb E_{P_W^\star}[\Delta_Q(W)]\):
\begin{align*}
\Psi(P_W,Q)-\Psi(P_W^\star,Q^\star)
&=
\left\{
\mathbb E_{P_W}[\Delta_Q(W)]
-
\mathbb E_{P_W^\star}[\Delta_Q(W)]
\right\} \\
&\quad+
\mathbb E_{P_W^\star}
[
\Delta_Q(W)-\Delta^\star(W)
].
\end{align*}
Taking absolute values and applying the triangle inequality gives
\[
|\Psi(P_W,Q)-\Psi(P_W^\star,Q^\star)|
\le
\left|
\int \Delta_Q(w)\,d(P_W-P_W^\star)(w)
\right|
+
\mathbb E_{P_W^\star}
\left[
|\Delta_Q(W)-\Delta^\star(W)|
\right].
\]
Since \(Y\in[0,1]\), we have \(|\Delta_Q(w)|\le 1\). Therefore,
\[
\left|
\int \Delta_Q(w)\,d(P_W-P_W^\star)(w)
\right|
\le
2\,\mathrm{TV}(P_W,P_W^\star),
\]
where \(\mathrm{TV}(P,Q)=\sup_A |P(A)-Q(A)|\). The second term is
\[
\mathbb E_{P_W^\star}
\left[
|\Delta_Q(W)-\Delta^\star(W)|
\right]
=
\|\Delta_Q-\Delta^\star\|_{L_1(P_W^\star)}.
\]
Combining the two bounds gives
\[
|\Psi(P_W,Q)-\Psi(P_W^\star,Q^\star)|
\le
2\,\mathrm{TV}(P_W,P_W^\star)
+
\|\Delta_Q-\Delta^\star\|_{L_1(P_W^\star)}.
\]
\end{proof}

\begin{thm}[Joint reconstruction induces an ATE-relevant tradeoff]
\label{thm:joint_reconstruction_tradeoff}
Let $W=(W_1,\ldots,W_d)$ denote the $d$ covariate coordinates in a
tabular row, let $A\in\{0,1\}$ denote treatment, and let
$Y\in\{0,1\}$. For a model $f$, define
\[
Q_f(a,w):=\mathbb P_f(Y=1\mid A=a,W=w),
\qquad
\Delta_f(w):=Q_f(1,w)-Q_f(0,w),
\]
and let $Q^\star$ and $\Delta^\star$ denote the corresponding truth.

Suppose the reconstruction objective is conditional on the treatment
coordinate $A$, so that $A$ is treated as observed context rather than as
a reconstructed component. Suppose further that the objective gives equal
weight to the $d$ covariate coordinates and one outcome component. Let
\[
L_W(f)
:=
\frac{1}{d+1}\sum_{j=1}^{d} L_{W_j}(f)
\]
denote the total weighted contribution of the covariate coordinates to the
conditional joint loss, and let $L_Y(f)$ denote the conditional outcome
loss for $Y\mid A,W$. Then
\[
L_{\mathrm{joint}}(f)
=
L_W(f)+\frac{1}{d+1}L_Y(f).
\]

For any reference model $g$,
\[
L_Y(f)-L_Y(g)
=
(d+1)\Bigl[
L_{\mathrm{joint}}(f)-L_{\mathrm{joint}}(g)
+
L_W(g)-L_W(f)
\Bigr].
\]
Thus, outcome-loss improvement inside the conditional joint reconstruction
objective is traded against covariate-reconstruction loss with a factor
$d+1$.

Assume further that $L_Y$ is the balanced Bernoulli KL loss. Then
\[
\|\Delta_f-\Delta^\star\|_{L_2(P_W^\star)}
\le
2\sqrt{L_Y(f)}.
\]
Consequently, if $L_{\mathrm{joint}}(f)\le \varepsilon$ and
$L_W(f)\ge 0$, then
\[
\|\Delta_f-\Delta^\star\|_{L_2(P_W^\star)}
\le
2\sqrt{(d+1)\varepsilon}.
\]
By contrast, if a hybrid outcome model $f_{\mathrm{hyb}}$ is fit directly
to $L_Y$ and satisfies
\[
L_Y(f_{\mathrm{hyb}})\le \varepsilon_Y,
\]
then
\[
\|\Delta_{f_{\mathrm{hyb}}}-\Delta^\star\|_{L_2(P_W^\star)}
\le
2\sqrt{\varepsilon_Y}.
\]
\end{thm}

\begin{proof}[Proof for \textbf{\Cref{thm:joint_reconstruction_tradeoff}}]
Because the reconstruction objective conditions on $A$ rather than
reconstructing it, the equal-weight objective contains $d+1$ reconstructed
components: the $d$ covariate coordinates and the outcome. By the definition
of $L_W$ as the total weighted contribution of the $d$ covariate
coordinates,
\[
L_{\mathrm{joint}}(f)
=
L_W(f)+\frac{1}{d+1}L_Y(f).
\]
Therefore, for any two models $f$ and $g$,
\[
L_{\mathrm{joint}}(f)-L_{\mathrm{joint}}(g)
=
\{L_W(f)-L_W(g)\}
+
\frac{1}{d+1}\{L_Y(f)-L_Y(g)\}.
\]
Rearranging gives
\[
\frac{1}{d+1}\{L_Y(f)-L_Y(g)\}
=
L_{\mathrm{joint}}(f)-L_{\mathrm{joint}}(g)
-
\{L_W(f)-L_W(g)\}.
\]
Multiplying both sides by $d+1$ yields
\[
L_Y(f)-L_Y(g)
=
(d+1)
\Bigl[
L_{\mathrm{joint}}(f)-L_{\mathrm{joint}}(g)
+
L_W(g)-L_W(f)
\Bigr].
\]
This is the exact tradeoff identity. It shows that a decrease in
covariate-reconstruction loss can compensate for an increase in outcome
loss inside the joint objective, and that converting joint-loss control
into outcome-loss control incurs a factor $d+1$.

As a special case, if $L_W(f)\ge 0$ and
\[
L_{\mathrm{joint}}(f)\le \varepsilon,
\]
then
\[
\frac{1}{d+1}L_Y(f)
\le
L_{\mathrm{joint}}(f)
\le
\varepsilon.
\]
Hence,
\[
L_Y(f)\le (d+1)\varepsilon.
\]

Now assume that $L_Y$ is the balanced Bernoulli KL loss:
\[
L_Y(f)
=
\mathbb E_{P_W^\star}
\sum_{a\in\{0,1\}}
\mathrm{KL}\!\left(
\mathrm{Bern}(Q^\star(a,W))
\,\middle\|\,
\mathrm{Bern}(Q_f(a,W))
\right).
\]
By Pinsker's inequality, for every $a\in\{0,1\}$ and every $w$,
\[
|Q_f(a,w)-Q^\star(a,w)|^2
\le
\frac{1}{2}
\mathrm{KL}\!\left(
\mathrm{Bern}(Q^\star(a,w))
\,\middle\|\,
\mathrm{Bern}(Q_f(a,w))
\right).
\]
Integrating over $P_W^\star$ gives
\[
\|Q_f(a,\cdot)-Q^\star(a,\cdot)\|_{L_2(P_W^\star)}^2
\le
\frac{1}{2}
\mathbb E_{P_W^\star}
\mathrm{KL}\!\left(
\mathrm{Bern}(Q^\star(a,W))
\,\middle\|\,
\mathrm{Bern}(Q_f(a,W))
\right)
\le
\frac{1}{2}L_Y(f).
\]
Therefore,
\[
\|Q_f(a,\cdot)-Q^\star(a,\cdot)\|_{L_2(P_W^\star)}
\le
\sqrt{\frac{L_Y(f)}{2}}.
\]

Because
\[
\Delta_f-\Delta^\star
=
\{Q_f(1,\cdot)-Q^\star(1,\cdot)\}
-
\{Q_f(0,\cdot)-Q^\star(0,\cdot)\},
\]
the triangle inequality gives
\begin{align*}
\|\Delta_f-\Delta^\star\|_{L_2(P_W^\star)}
&\le
\|Q_f(1,\cdot)-Q^\star(1,\cdot)\|_{L_2(P_W^\star)}
\\
&\quad+
\|Q_f(0,\cdot)-Q^\star(0,\cdot)\|_{L_2(P_W^\star)}
\\
&\le
2\sqrt{\frac{L_Y(f)}{2}}
\\
&\le
2\sqrt{L_Y(f)}.
\end{align*}
The final constant is loose but sufficient.

Combining this bound with
\[
L_Y(f)\le (d+1)\varepsilon
\]
gives
\[
\|\Delta_f-\Delta^\star\|_{L_2(P_W^\star)}
\le
2\sqrt{(d+1)\varepsilon}.
\]

For the hybrid outcome model, the assumption
\[
L_Y(f_{\mathrm{hyb}})\le \varepsilon_Y
\]
is a direct conditional-outcome loss certificate. Applying the same
Pinsker and triangle-inequality argument gives
\[
\|\Delta_{f_{\mathrm{hyb}}}-\Delta^\star\|_{L_2(P_W^\star)}
\le
2\sqrt{\varepsilon_Y}.
\]

Thus, joint reconstruction controls the ATE-relevant outcome contrast only
through a $(d+1)$-inflated certificate, whereas direct hybrid outcome
fitting controls it without this dilution.
\end{proof}

\begin{prop}[When synthetic overlap support helps]
\label{prop:overlap}
Let
\[
\psi_0=\int \tau_0(w)\,d\mu_0(w), \qquad \tau_0(w)=Q_0(1,w)-Q_0(0,w),
\]
be the target ATE under the target covariate law $\mu_0$. Let
\[
\tau_{\mathrm{orig}}(w)=Q_{\mathrm{orig}}(1,w)-Q_{\mathrm{orig}}(0,w), \qquad \tau_{\mathrm{aug}}(w)=Q_{\mathrm{aug}}(1,w)-Q_{\mathrm{aug}}(0,w),
\]
and define
\[
\psi_{\mathrm{orig}}=\int \tau_{\mathrm{orig}}(w)\,d\mu_0(w), \qquad
\psi_{\mathrm{aug}}=\int \tau_{\mathrm{aug}}(w)\,d\mu_{\mathrm{aug}}(w),
\]
where $\mu_{\mathrm{aug}}$ is the covariate law induced by synthetic augmentation. Then
\[
\psi_{\mathrm{orig}}-\psi_0
=
\int (\tau_{\mathrm{orig}}-\tau_0)\,d\mu_0,
\]
and
\[
\psi_{\mathrm{aug}}-\psi_0
=
\int (\tau_{\mathrm{aug}}-\tau_0)\,d\mu_{\mathrm{aug}}
+
\int \tau_0\,d(\mu_{\mathrm{aug}}-\mu_0).
\]
Therefore, synthetic overlap support improves absolute error,
\[
|\psi_{\mathrm{aug}}-\psi_0|<|\psi_{\mathrm{orig}}-\psi_0|,
\]
if and only if
\[
\left|
\int (\tau_{\mathrm{aug}}-\tau_0)\,d\mu_{\mathrm{aug}}
+
\int \tau_0\,d(\mu_{\mathrm{aug}}-\mu_0)
\right|
<
\left|
\int (\tau_{\mathrm{orig}}-\tau_0)\,d\mu_0
\right|.
\]
In particular, the key comparison is between two terms: augmentation can reduce error by improving the conditional effect estimate $\tau$, but it can also introduce a distribution-shift term through $\mu_{\mathrm{aug}}-\mu_0$.
\end{prop}

\begin{proof}[Proof for \textbf{\Cref{prop:overlap}}]
Start from the definitions of the augmented estimator and the target:
\begin{align*}
\psi_{\mathrm{aug}}-\psi_0
&=
\int \tau_{\mathrm{aug}}(w)\, d\mu_{\mathrm{aug}}(w)
-
\int \tau_0(w)\, d\mu_0(w).
\end{align*}
Insert and subtract the same intermediate quantity $\int \tau_0\, d\mu_{\mathrm{aug}}$:
\begin{align*}
\psi_{\mathrm{aug}}-\psi_0
&=
\left(
\int \tau_{\mathrm{aug}}\, d\mu_{\mathrm{aug}}
-
\int \tau_0\, d\mu_{\mathrm{aug}}
\right)
+
\left(
\int \tau_0\, d\mu_{\mathrm{aug}}
-
\int \tau_0\, d\mu_0
\right).
\end{align*}
By linearity of integration,
\begin{align*}
\int \tau_{\mathrm{aug}}\, d\mu_{\mathrm{aug}}
-
\int \tau_0\, d\mu_{\mathrm{aug}}
&=
\int (\tau_{\mathrm{aug}}-\tau_0)\, d\mu_{\mathrm{aug}},
\end{align*}
and, by the definition of integration against a signed measure,
\begin{align*}
\int \tau_0\, d\mu_{\mathrm{aug}}
-
\int \tau_0\, d\mu_0
&=
\int \tau_0\, d(\mu_{\mathrm{aug}}-\mu_0).
\end{align*}
Therefore,
\begin{align*}
\psi_{\mathrm{aug}}-\psi_0
=
\int (\tau_{\mathrm{aug}}-\tau_0)\, d\mu_{\mathrm{aug}}
+
\int \tau_0\, d(\mu_{\mathrm{aug}}-\mu_0).
\end{align*}

For the original estimator,
\begin{align*}
\psi_{\mathrm{orig}}-\psi_0
&=
\int \tau_{\mathrm{orig}}\, d\mu_0
-
\int \tau_0\, d\mu_0 \\
&=
\int (\tau_{\mathrm{orig}}-\tau_0)\, d\mu_0.
\end{align*}

Now synthetic overlap support is beneficial exactly when
\begin{align*}
|\psi_{\mathrm{aug}}-\psi_0| < |\psi_{\mathrm{orig}}-\psi_0|.
\end{align*}
Substituting the two exact identities above yields
\begin{align*}
\left|
\int (\tau_{\mathrm{aug}}-\tau_0)\, d\mu_{\mathrm{aug}}
+
\int \tau_0\, d(\mu_{\mathrm{aug}}-\mu_0)
\right|
<
\left|
\int (\tau_{\mathrm{orig}}-\tau_0)\, d\mu_0
\right|.
\end{align*}
This is precisely the claimed equivalence. The proof is complete.
\end{proof}

\begin{thm}[Exact decomposition of prediction loss]
\label{thm:exact_lq_cate_ate_decomposition}
Let $A\in\{0,1\}$ and let
\[
\pi(w)=\mathbb P(A=1\mid W=w),
\qquad 0<\pi(W)<1 \quad \text{almost surely}.
\]
Let
\[
Q(a,w)=\mathbb E[Y\mid A=a,W=w],
\qquad a\in\{0,1\},
\]
and define
\[
\tau(w)=Q(1,w)-Q(0,w),
\]
and
\[
m_\pi(w)=\mathbb E[Y\mid W=w]
=
\pi(w)Q(1,w)+\{1-\pi(w)\}Q(0,w).
\]
Then
\[
Q(A,W)=m_\pi(W)+\{A-\pi(W)\}\tau(W).
\]

For any fitted outcome regression $\widehat Q$, define
\[
\widehat\tau(w)=\widehat Q(1,w)-\widehat Q(0,w),
\]
and
\[
\widehat m_\pi(w)
=
\pi(w)\widehat Q(1,w)+\{1-\pi(w)\}\widehat Q(0,w).
\]
Then
\[
\widehat Q(A,W)
=
\widehat m_\pi(W)+\{A-\pi(W)\}\widehat\tau(W).
\]
Define
\[
\Delta_m(W)=\widehat m_\pi(W)-m_\pi(W),
\qquad
\Delta_\tau(W)=\widehat\tau(W)-\tau(W),
\]
and let
\[
\omega(W)=\pi(W)\{1-\pi(W)\},
\qquad
\bar\omega=\mathbb E_W[\omega(W)].
\]

Let the outcome-prediction loss be
\[
\mathcal L_Q
=
\mathbb E_{A,W}
\left[
\{\widehat Q(A,W)-Q(A,W)\}^2
\right],
\]
the CATE loss be
\[
\mathcal L_\tau
=
\mathbb E_W[\Delta_\tau(W)^2],
\]
and the ATE loss be
\[
\mathcal L_\psi
=
(\widehat\psi-\psi)^2,
\qquad
\psi=\mathbb E_W[\tau(W)],
\quad
\widehat\psi=\mathbb E_W[\widehat\tau(W)].
\]
Then the prediction loss has the exact decomposition
\[
\mathcal L_Q
=
\mathbb E_W[\Delta_m(W)^2]
+
\mathbb E_W[\omega(W)\Delta_\tau(W)^2].
\]
Equivalently, in terms of the CATE loss,
\[
\mathcal L_Q
=
\mathbb E_W[\Delta_m(W)^2]
+
\bar\omega\,\mathcal L_\tau
+
\mathrm{Cov}_W\!\left(\omega(W),\Delta_\tau(W)^2\right).
\]
Since
\[
\mathcal L_\tau
=
\mathcal L_\psi
+
\mathrm{Var}_W\{\Delta_\tau(W)\},
\]
we also have the exact ATE decomposition
\[
\mathcal L_Q
=
\mathbb E_W[\Delta_m(W)^2]
+
\bar\omega\,\mathcal L_\psi
+
\bar\omega\,\mathrm{Var}_W\{\Delta_\tau(W)\}
+
\mathrm{Cov}_W\!\left(\omega(W),\Delta_\tau(W)^2\right).
\]
\end{thm}

\begin{proof}
First, by the law of total expectation,
\[
m_\pi(w)
=
\mathbb E[Y\mid W=w]
=
\pi(w)Q(1,w)+\{1-\pi(w)\}Q(0,w).
\]
Since
\[
\tau(w)=Q(1,w)-Q(0,w),
\]
we can solve for the two conditional means:
\[
Q(1,w)=m_\pi(w)+\{1-\pi(w)\}\tau(w),
\]
and
\[
Q(0,w)=m_\pi(w)-\pi(w)\tau(w).
\]
Therefore, for $A\in\{0,1\}$,
\[
Q(A,W)
=
m_\pi(W)+\{A-\pi(W)\}\tau(W).
\]

The same algebra applies to $\widehat Q$. By definition,
\[
\widehat\tau(w)=\widehat Q(1,w)-\widehat Q(0,w),
\]
and
\[
\widehat m_\pi(w)
=
\pi(w)\widehat Q(1,w)+\{1-\pi(w)\}\widehat Q(0,w).
\]
Thus
\[
\widehat Q(1,w)
=
\widehat m_\pi(w)+\{1-\pi(w)\}\widehat\tau(w),
\]
and
\[
\widehat Q(0,w)
=
\widehat m_\pi(w)-\pi(w)\widehat\tau(w),
\]
so
\[
\widehat Q(A,W)
=
\widehat m_\pi(W)+\{A-\pi(W)\}\widehat\tau(W).
\]

Subtracting the true and fitted regressions gives
\[
\widehat Q(A,W)-Q(A,W)
=
\Delta_m(W)+\{A-\pi(W)\}\Delta_\tau(W).
\]
Therefore,
\[
\mathcal L_Q
=
\mathbb E_{A,W}
\left[
\left\{
\Delta_m(W)+\{A-\pi(W)\}\Delta_\tau(W)
\right\}^2
\right].
\]
Expanding the square,
\[
\mathcal L_Q
=
\mathbb E_W[\Delta_m(W)^2]
+
2\mathbb E_{A,W}
\left[
\Delta_m(W)\{A-\pi(W)\}\Delta_\tau(W)
\right]
+
\mathbb E_{A,W}
\left[
\{A-\pi(W)\}^2\Delta_\tau(W)^2
\right].
\]
The cross term is zero because
\[
\mathbb E[A-\pi(W)\mid W]=0.
\]
Also,
\[
\mathbb E[\{A-\pi(W)\}^2\mid W]
=
\mathrm{Var}(A\mid W)
=
\pi(W)\{1-\pi(W)\}
=
\omega(W).
\]
Hence
\[
\mathcal L_Q
=
\mathbb E_W[\Delta_m(W)^2]
+
\mathbb E_W[\omega(W)\Delta_\tau(W)^2].
\]

Now write
\[
\mathbb E_W[\omega(W)\Delta_\tau(W)^2]
=
\bar\omega\,\mathbb E_W[\Delta_\tau(W)^2]
+
\mathrm{Cov}_W\!\left(\omega(W),\Delta_\tau(W)^2\right).
\]
Since
\[
\mathcal L_\tau
=
\mathbb E_W[\Delta_\tau(W)^2],
\]
we obtain
\[
\mathcal L_Q
=
\mathbb E_W[\Delta_m(W)^2]
+
\bar\omega\,\mathcal L_\tau
+
\mathrm{Cov}_W\!\left(\omega(W),\Delta_\tau(W)^2\right).
\]

Finally,
\[
\widehat\psi-\psi
=
\mathbb E_W[\widehat\tau(W)-\tau(W)]
=
\mathbb E_W[\Delta_\tau(W)].
\]
Therefore,
\[
\mathcal L_\psi
=
(\widehat\psi-\psi)^2
=
\left\{\mathbb E_W[\Delta_\tau(W)]\right\}^2.
\]
By the mean--variance decomposition,
\[
\mathbb E_W[\Delta_\tau(W)^2]
=
\left\{\mathbb E_W[\Delta_\tau(W)]\right\}^2
+
\mathrm{Var}_W\{\Delta_\tau(W)\}.
\]
Thus
\[
\mathcal L_\tau
=
\mathcal L_\psi
+
\mathrm{Var}_W\{\Delta_\tau(W)\}.
\]
Substituting this into the CATE-loss decomposition gives
\[
\mathcal L_Q
=
\mathbb E_W[\Delta_m(W)^2]
+
\bar\omega\,\mathcal L_\psi
+
\bar\omega\,\mathrm{Var}_W\{\Delta_\tau(W)\}
+
\mathrm{Cov}_W\!\left(\omega(W),\Delta_\tau(W)^2\right).
\]
This proves the theorem.
\end{proof}

\Cref{thm:exact_lq_cate_ate_decomposition} shows that outcome prediction and causal-effect preservation
optimize different components of the same regression. The prediction
loss decomposes into error in the prognostic component
\(m_\pi(W)=\mathbb E[Y\mid W]\) and an overlap-weighted error in the
treatment-effect component \(\tau(W)\). Thus \(\mathcal L_Q\) is not a pure
causal loss: it prioritizes treatment-effect accuracy most strongly in regions
where both treatment arms are well represented.

In a balanced randomized trial, \(\pi(W)=1/2\), so
\(\omega(W)=1/4\) and the covariance term vanishes. The decomposition reduces to
\[
\mathcal L_Q
=
\mathbb E_W[\Delta_m(W)^2]
+
\frac14\mathcal L_\tau
=
\mathbb E_W[\Delta_m(W)^2]
+
\frac14\mathcal L_\psi
+
\frac14\mathrm{Var}_W\{\Delta_\tau(W)\}.
\]
Thus even under balanced randomization, treatment-effect error enters factual
prediction loss only through a factor \(1/4\).

In an imbalanced or limited-overlap setting,
\(\omega(W)=\pi(W)\{1-\pi(W)\}\) can be much smaller than \(1/4\). Moreover, if
treatment-effect errors are concentrated in low-overlap regions, then
\[
\mathrm{Cov}_W\!\left(\omega(W),\Delta_\tau(W)^2\right)<0,
\]
so large CATE errors may contribute little to prediction loss. If those
errors also have nonzero average, the ATE error may likewise be poorly reflected
by \(\mathcal L_Q\). This explains why a model can predict outcomes
well while still failing to preserve causal effects.

\section*{A Decomposition for Block-Specific Continuation Prediction}

Let
\[
(C,S,Z)
\]
denote a sequence split into three parts: \(C\) is a prefix, \(S\in\mathcal S\)
is a token block following the prefix, and \(Z\) is the downstream continuation.
We study how the continuation law varies with the intermediate block \(S\) while
the prefix \(C\) is held fixed. Both \(S\) and \(Z\) may have variable length;
EOS tokens are treated as part of the sequence law.

Let \(p^\star\) be the true law and \(p_\theta\) a fitted model. For each
\(c\in\mathcal C\) and \(s\in\mathcal S\), define
\[
\rho(s\mid c):=p^\star(S=s\mid C=c),
\]
\[
q_s(\cdot\mid c):=p^\star(Z\in\cdot\mid C=c,S=s),
\qquad
\widehat q_s(\cdot\mid c):=p_\theta(Z\in\cdot\mid C=c,S=s).
\]
Thus \(q_s(\cdot\mid c)\) is the true continuation law after prefix \(c\) and
block \(s\), while \(\widehat q_s(\cdot\mid c)\) is the corresponding fitted
law.

Define the prefix-marginal continuation laws
\[
m(\cdot\mid c)
:=
\sum_{s\in\mathcal S}\rho(s\mid c)\,q_s(\cdot\mid c)
=
p^\star(Z\in\cdot\mid C=c),
\]
and
\[
\widehat m(\cdot\mid c)
:=
\sum_{s\in\mathcal S}\rho(s\mid c)\,\widehat q_s(\cdot\mid c).
\]
The law \(m(\cdot\mid c)\) describes the continuation distribution after
observing the prefix \(C=c\), with the intermediate block \(S\) integrated out.
By contrast, the collection
\[
\{q_s(\cdot\mid c):s\in\mathcal S\}
\]
describes the block-specific continuation structure: how the downstream
continuation varies across different intermediate blocks.

Define the continuation negative log-likelihood
\[
\mathcal L_Z(\theta)
:=
\mathbb E_{C,S,Z\sim p^\star}\big[-\log p_\theta(Z\mid C,S)\big].
\]
Also define the true and fitted posteriors over \(S\) given \((C,Z)\):
\[
\eta(s\mid c,z)
:=
\frac{\rho(s\mid c)\,q_s(z\mid c)}{m(z\mid c)},
\qquad
\widehat\eta(s\mid c,z)
:=
\frac{\rho(s\mid c)\,\widehat q_s(z\mid c)}{\widehat m(z\mid c)}.
\]
The posterior \(\eta(s\mid c,z)\) measures how informative the observed
continuation \(z\) is about which intermediate block \(s\) preceded it, given
the prefix \(c\).

\begin{thm}[Exact decomposition of continuation prediction loss]
\label{thm:ntp_causal_decomposition}
The continuation negative log-likelihood satisfies
\[
\mathcal L_Z(\theta)
=
H^\star(Z\mid C,S)
+
\mathbb E_C\!\left[
D_{\mathrm{KL}}\!\bigl(m(\cdot\mid C)\,\|\,\widehat m(\cdot\mid C)\bigr)
\right]
+
\mathbb E_C\!\left[
\mathbb E_{Z\sim m(\cdot\mid C)}
D_{\mathrm{KL}}\!\bigl(\eta(\cdot\mid C,Z)\,\|\,\widehat\eta(\cdot\mid C,Z)\bigr)
\right].
\]
Equivalently, since
\[
H^\star(Z\mid C,S)=H^\star(Z\mid C)-I^\star(Z;S\mid C),
\]
we have
\[
\mathcal L_Z(\theta)
=
H^\star(Z\mid C)
-
I^\star(Z;S\mid C)
+
\mathbb E_C\!\left[
D_{\mathrm{KL}}\!\bigl(m(\cdot\mid C)\,\|\,\widehat m(\cdot\mid C)\bigr)
\right]
+
\mathbb E_C\!\left[
\mathbb E_{Z\sim m(\cdot\mid C)}
D_{\mathrm{KL}}\!\bigl(\eta(\cdot\mid C,Z)\,\|\,\widehat\eta(\cdot\mid C,Z)\bigr)
\right].
\]
\end{thm}

\begin{proof}
We begin from the definition of conditional cross-entropy. Since
\[
\mathcal L_Z(\theta)
=
\mathbb E_{C,S,Z\sim p^\star}\big[-\log p_\theta(Z\mid C,S)\big],
\]
we may write
\[
\mathcal L_Z(\theta)
=
H^\star(Z\mid C,S)
+
\mathbb E_{C,S}
D_{\mathrm{KL}}\!\bigl(
p^\star(Z\mid C,S)\,\|\,p_\theta(Z\mid C,S)
\bigr).
\]
Using the notation \(q_s(\cdot\mid c)=p^\star(Z\in\cdot\mid C=c,S=s)\) and
\(\widehat q_s(\cdot\mid c)=p_\theta(Z\in\cdot\mid C=c,S=s)\), the second term
becomes
\[
\mathbb E_C\!\left[
\sum_{s\in\mathcal S}\rho(s\mid C)\,
D_{\mathrm{KL}}\!\bigl(q_s(\cdot\mid C)\,\|\,\widehat q_s(\cdot\mid C)\bigr)
\right].
\]

Fix \(c\). Define two joint laws on \((S,Z)\) by
\[
P_c(s,z):=\rho(s\mid c)\,q_s(z\mid c),
\qquad
\widehat P_c(s,z):=\rho(s\mid c)\,\widehat q_s(z\mid c).
\]
Then, by direct expansion of the KL divergence,
\[
\begin{aligned}
D_{\mathrm{KL}}(P_c\,\|\,\widehat P_c)
&=
\sum_{s\in\mathcal S}\int
\rho(s\mid c)\,q_s(z\mid c)
\log\frac{\rho(s\mid c)\,q_s(z\mid c)}
{\rho(s\mid c)\,\widehat q_s(z\mid c)}
\,dz \\
&=
\sum_{s\in\mathcal S}\rho(s\mid c)
\int q_s(z\mid c)\log\frac{q_s(z\mid c)}{\widehat q_s(z\mid c)}\,dz \\
&=
\sum_{s\in\mathcal S}\rho(s\mid c)\,
D_{\mathrm{KL}}\!\bigl(q_s(\cdot\mid c)\,\|\,\widehat q_s(\cdot\mid c)\bigr).
\end{aligned}
\]
Therefore
\[
\mathcal L_Z(\theta)
=
H^\star(Z\mid C,S)
+
\mathbb E_C\!\left[
D_{\mathrm{KL}}(P_C\,\|\,\widehat P_C)
\right].
\]

We now apply the chain rule for KL to \(P_c\) and \(\widehat P_c\). Their
marginals on \(Z\) are
\[
P_c^Z(z)=\sum_{s\in\mathcal S}\rho(s\mid c)\,q_s(z\mid c)=m(z\mid c),
\]
and
\[
\widehat P_c^Z(z)=\sum_{s\in\mathcal S}\rho(s\mid c)\,\widehat q_s(z\mid c)=\widehat m(z\mid c).
\]
Their conditional laws of \(S\) given \(Z=z\) are
\[
P_c(s\mid z)=\eta(s\mid c,z),
\qquad
\widehat P_c(s\mid z)=\widehat\eta(s\mid c,z).
\]
Hence the chain rule yields
\[
D_{\mathrm{KL}}(P_c\,\|\,\widehat P_c)
=
D_{\mathrm{KL}}\!\bigl(m(\cdot\mid c)\,\|\,\widehat m(\cdot\mid c)\bigr)
+
\mathbb E_{Z\sim m(\cdot\mid c)}
D_{\mathrm{KL}}\!\bigl(\eta(\cdot\mid c,Z)\,\|\,\widehat\eta(\cdot\mid c,Z)\bigr).
\]
Substituting this identity into the previous display and then taking
expectation over \(C\) proves the first decomposition.

Finally, the identity
\[
H^\star(Z\mid C,S)=H^\star(Z\mid C)-I^\star(Z;S\mid C)
\]
gives the second decomposition.
\end{proof}

Theorem~\ref{thm:ntp_causal_decomposition} decomposes continuation prediction
into two model-dependent components. The first,
\[
D_{\mathrm{KL}}\!\bigl(m(\cdot\mid C)\,\|\,\widehat m(\cdot\mid C)\bigr),
\]
measures discrepancy in the prefix-marginal continuation law. This term concerns
the average continuation distribution after observing \(C\), with the
intermediate block \(S\) integrated out. The second,
\[
\mathbb E_{Z\sim m(\cdot\mid C)}
D_{\mathrm{KL}}\!\bigl(\eta(\cdot\mid C,Z)\,\|\,\widehat\eta(\cdot\mid C,Z)\bigr),
\]
measures discrepancy in the posterior distribution over the intermediate block
\(S\) induced by the observed pair \((C,Z)\). We refer to this as a posterior
block-distinguishability term: it evaluates whether the fitted continuation laws
\(\{\widehat q_s(\cdot\mid C):s\in\mathcal S\}\) preserve enough
block-specific structure for the continuation \(Z\) to distinguish among possible
intermediate blocks.

This decomposition highlights a separation between fitting the averaged
prefix-to-continuation law and preserving block-specific continuation structure.
The marginal law
\[
m(\cdot\mid c)=\sum_{s\in\mathcal S}\rho(s\mid c)q_s(\cdot\mid c)
\]
may be broad, multimodal, or dominated by high-probability continuation
patterns. This is especially natural when the prefix \(C\) is distant from the
downstream continuation \(Z\). In such cases, a model can improve continuation
likelihood by approximating the averaged law \(m(\cdot\mid C)\), while retaining
only a coarse representation of how \(Z\) depends on the intermediate block
\(S\).

The posterior block-distinguishability term makes the role of \(S\) explicit.
For a fixed prefix \(c\),
\[
\eta(s\mid c,z)
=
\frac{\rho(s\mid c)q_s(z\mid c)}{m(z\mid c)}
\]
summarizes how much the continuation \(z\) reveals about the block \(s\) that
preceded it. Accurate modeling of this posterior requires the fitted laws
\(\widehat q_s(\cdot\mid c)\) to retain the relative differences among the
block-specific continuation distributions. Thus, although \(S\) is adjacent to
the continuation \(Z\) and may contain more immediate information about it than
the distant prefix \(C\), ordinary continuation likelihood does not give this
local dependence a separate objective. The information carried by \(S\) is
penalized only through its contribution to the observed joint law of \((S,Z)\)
given \(C\).

This distinction is important under imbalance. The block-specific continuation
laws enter the likelihood through the observed block weights
\[
\rho(s\mid c)=p^\star(S=s\mid C=c).
\]
Consequently, dependencies associated with low-probability blocks make a smaller
contribution to the objective, even when those dependencies are structurally
important for downstream behavior. In this sense, sequential proximity alone
does not guarantee reliable learning of the \(S\)-specific continuation law:
a block may be close to \(Z\), but if it is rare under \(\rho(\cdot\mid C)\), the
training objective provides weak signal for its associated continuation
structure.

The same issue can arise for long or compositionally specific blocks. Since
\[
\rho(s\mid c)
=
p^\star(s\mid c)
=
\prod_{j=1}^{|s|}p^\star(s_j\mid c,s_{<j}),
\]
the probability of an entire block may be small even when its individual tokens
are not extremely rare. Longer, more structured, or multi-step intermediate
blocks can therefore receive limited effective weight in the continuation
objective. This creates a structural disadvantage for learning dependencies that
are expressed through extended blocks rather than through short, high-frequency
patterns.

The decomposition also shows why standard continuation loss can be an incomplete
diagnostic for block-specific structure. A fitted model may match the
prefix-marginal law \(m(\cdot\mid C)\) well while failing to reproduce the
posterior block-distinguishability encoded by \(\eta(\cdot\mid C,Z)\). In that
case, the model predicts the averaged continuation distribution accurately but
does not faithfully preserve which intermediate blocks give rise to which
downstream continuations. Thus good likelihood need not imply accurate modeling
of block-specific dependencies.

This provides a sequence-level analogue of the causal prediction problem without
requiring a literal treatment-effect contrast. In the causal setting, prediction
loss may emphasize dominant observed outcome patterns while giving weak signal
for treatment-specific structure in low-overlap regions. Here, continuation
likelihood may emphasize the prefix-marginal continuation law while giving weak
signal for block-specific continuation structure under an imbalanced block
distribution. When downstream use depends on low-frequency or compositionally
specific dependencies, reweighting, balancing, or otherwise modifying the
training distribution may better align the next-token objective with the
structure one wants the model to preserve.

\end{document}